\newif\ifshaphered

\ifdefined\isshaphered
\shapheredtrue
\fi
\documentclass[sigconf]{acmart}

\AtBeginDocument{%
  }

\setcopyright{acmlicensed}
\copyrightyear{2018}
\acmYear{2018}
\acmDOI{XXXXXXX.XXXXXXX}

\acmConference[Conference acronym 'XX]{Make sure to enter the correct
  conference title from your rights confirmation emai}{June 03--05,
  2018}{Woodstock, NY}
\acmISBN{978-1-4503-XXXX-X/18/06}

\usepackage{algorithm}
\usepackage{algorithmic}
\usepackage{tabularx}
\usepackage{diagbox}
\usepackage{array}
\usepackage{enumitem}
\usepackage{hyperref}
\usepackage{multirow}
\usepackage{makecell}

\renewcommand{\eqref}[1]{\mbox{Equation~(\ref{#1})}}

\newtheorem{definition}{Definition}
\newtheorem{theorem}{Theorem}

\ifodd 1
\newcommand{\fengyu}[1]{{\color{blue}(fengyu: #1)}}
\else
\newcommand{\fengyu}[1]{#}
\fi

\ifshaphered
\usepackage{longtable}
\newcommand\revision[1]{\textcolor{blue}{#1}}

\else
\newcommand\revision[1]{#1}

\fi
\begin{document}

\title{HeteroFedSyn: Differentially Private Tabular Data Synthesis for Heterogeneous Federated Settings}


\author{Xiaochen Li}
\affiliation{%
  \institution{UNC Greensboro}
  \country{USA}}
\email{x_li12@uncg.edu}

\author{Fengyu Gao}
\affiliation{%
  \institution{University of Virginia}
  \country{USA}}
\email{wan6jj@virginia.edu}

\author{Xizixiang Wei}
\affiliation{%
  \institution{University of Virginia}
  \country{USA}}
\email{xw8cw@virginia.edu}

\author{Tianhao Wang}
\affiliation{%
  \institution{University of Virginia}
  \country{USA}}
\email{tianhao@virginia.edu}

\author{Cong Shen}
\affiliation{%
  \institution{University of Virginia}
  \country{USA}}
\email{cong@virginia.edu}

\author{Jing Yang}
\affiliation{%
  \institution{University of Virginia}
  \country{USA}}
\email{yangjing@virginia.edu}

\renewcommand{\shortauthors}{Li et al.}

\begin{abstract}
Traditional Differential Privacy (DP) mechanisms are typically tailored to specific analysis tasks, which limits the reusability of protected data. DP tabular data synthesis overcomes this by generating synthetic datasets that can be shared for arbitrary downstream tasks. However, existing synthesis methods predominantly assume centralized or local settings and overlook the more practical horizontal federated scenario. Naïvely synthesizing data locally or perturbing individual records either produces biased mixtures or introduces excessive noise, especially under heterogeneous data distributions across participants.

We propose HeteroFedSyn, the first DP tabular data synthesis framework designed specifically for the horizontal federated setting. Built upon the PrivSyn paradigm of 2-way marginal–based synthesis, HeteroFedSyn introduces three key innovations for distributed marginal selection: (i) an $l_2$-based dependency metric with random projection for noise-efficient correlation measurement, (ii) an unbiased estimator to correct multiplicative noise, and (iii) an adaptive selection strategy that dynamically updates dependency scores to avoid redundancy.
Extensive experiments on range queries, Wasserstein fidelity, and machine learning tasks show that, despite the increased noise inherent to federated execution, HeteroFedSyn achieves utility comparable to centralized synthesis. Our code is open-sourced via the link\footnote{https://github.com/XiaochenLi-w/Federated-Tabular-Data-Synthesis-Framework}.
\footnote{This work has been accepted by SIGMOD 2026.}
\end{abstract}

\maketitle

\pagestyle{plain}
\setcounter{page}{1}
\pagenumbering{arabic}

\section{Introduction}
Protecting data privacy while leveraging sensitive datasets for data mining tasks has long been a fundamental research challenge.
To solve this challenge, Differential Privacy (DP) is proposed~\cite{DBLP:conf/tcc/DworkMNS06}, which provides strong privacy guarantees and has been widely adopted across diverse applications, such as Census~\cite{abowd2018us}, Healthcare~\cite{ziller2021medical}, and GPS~\cite{DBLP:conf/ccs/AndresBCP13}.
Traditional DP achieves privacy by injecting noise into (i) the dataset itself, (ii) the outputs of data analysis, or (iii) intermediate parameters.
These approaches have been extensively studied~\cite{li2025spas, feng2024dpi, dong2023continual, dwork2010differential, erlingsson2014rappor, gao2025data, DBLP:conf/ndss/LiQ00F0025}, covering tasks from range queries~\cite{kulkarni2019answering, li2014data} to generative model training~\cite{li2025easy, jiang2023functional}.
Although these techniques are considered mature from a research perspective, several concerns persist in real-world deployment.

A key limitation is that DP protection via noise addition is typically tied to a specific task. 
First, the scale of DP noise depends on data sensitivity, which is usually task-dependent. 
While domain-level sensitivity is also possible, it often introduces excessive noise. 
Second, noise mechanisms often require customization for each analysis pipeline, especially when injecting noise into intermediate results. 
This reduces usability and increases analysis cost.

To enable task-agnostic use of DP-protected data, researchers have proposed DP tabular data synthesis methods.
These methods generate synthetic datasets from sensitive data and allow release for arbitrary downstream tasks.
They typically operate in two steps: (1) extracting noisy statistics from the original data, and (2) sampling synthetic records that match them.
Existing methods fall into graphical-model–based~\cite{DBLP:conf/sigmod/ZhangCPSX14, DBLP:journals/pvldb/BindschaedlerSG17, DBLP:conf/uss/00010LH0H0021, mckenna2022aim} and ML-based approaches~\cite{xie2018differentially, torkzadehmahani2019dp, jordon2018pate, truda2023generating, sattarov2024differentially, tran2024differentially}, with recent surveys providing comprehensive summaries~\cite{chen2025benchmarking, yang2024tabular, hu2024sok}.
Although not always optimal for a specific task, these methods eliminate repetitive access to raw data, avoid task-specific DP mechanisms, and maintain compatibility with existing analysis tools.
As a result, DP data synthesis substantially reduces the cost of data sharing.

However, most DP tabular data synthesis methods assume a centralized data setting, where all records are stored on a single server.
A separate line of work explores a local DP setting~\cite{li2024distributed}, where each user injects noise into their own record before sharing it.
While both settings are useful, they overlook a more realistic and common scenario: multiple organizations holding disjoint subsets of data with the same attributes and wishing to collaborate.
Examples include hospitals jointly analyzing regional disease trends and schools aggregating student statistics for resource planning.
These cases fall under the horizontal federated setting, where data are distributed across parties but share the same attributes.

Two intuitive solutions fall short in this setting.
(1) Each party could independently run a DP synthesis algorithm and share local synthetic data.
However, real-world data distributions are often heterogeneous across institutions, e.g., hospitals with different specialties or banks with distinct client bases, and the merged data would form a biased and inconsistent mixture.
(2) Parties could add LDP noise to individual records and then share them, but this renders further DP synthesis unnecessary and introduces variance that scales quadratically with the dataset size, severely degrading utility.
Therefore, instead of sharing raw data or locally perturbed records, it is necessary to collaboratively exchange dataset statistics to synthesize a global DP dataset.

We follow the widely adopted PrivSyn paradigm~\cite{DBLP:conf/uss/00010LH0H0021}, synthesizing data using 2-way marginals.
However, even sharing local 2-way marginals can still lead to high noise under a limited privacy budget.
Specifically, for a dataset with $d$ attributes, the variance of noise in each marginal is $O(d^4/(4\epsilon^2))$.
Therefore, only a subset of informative marginals should be selected to guide data synthesis.
While PrivSyn provides a centralized greedy selection method, it cannot be directly applied when data are distributed and inaccessible.

To address this gap, we propose HeteroFedSyn, a framework for DP tabular data synthesis in the horizontal federated setting, with three key innovations:
(a) \emph{Attribute Dependency Metric and Marginal Compression.}
We define an $l_2$ distance metric $\text{InDif2}_{a,b}$ to measure dependencies between any two attributes $a$ and $b$.
To reduce noise and communication overhead, we apply random projection to compress the 2-way marginals while preserving dependency signals $\text{InDif2}_{a,b}$.
Specifically, for two attributes with domains of size $d_a$ and $d_b$, the length of their $2$-way marginal is $d_a\times d_b$.
It can be compressed to length $k$ using random projection, while $k<<d_a,\ d_b$.
(b) \emph{Unbiased Estimation over Noisy marginals.}
Computing $\text{InDif2}_{a,b}$ involves multiplicative operations over noisy marginals, which makes debiasing particularly challenging.
We provide a rigorous mathematical procedure for obtaining an unbiased estimate of $\text{InDif2}_{a,b}$ from the compressed noisy marginals.
(c) \emph{Adaptive Marginal Selection.}
We observe that selecting important 2-way marginals solely based on attribute dependency can be suboptimal, as it ignores the overlap between selected and unselected marginals. 
For example, once the marginals for attribute pairs (a, b) and (a, c) are selected, the correlation between b and c is already implicitly constrained. Therefore, selecting the (b, c) marginal next may be redundant, even if (b, c) has a high $\text{InDif2}_{b,c}$.
In this case, the privacy budget could be better spent on covering additional attributes.
Therefore, we introduce an adaptive marginal selection mechanism that updates $\text{InDif2}_{a,b}$ adaptively during the selection process to avoid redundancy and maximize coverage under a fixed privacy budget.

Finally, we conduct extensive experiments on a variety of downstream tasks, including range queries, Wasserstein-based fidelity, and three machine learning models (Random Forest, MLP, and XGBoost).
The results show that although the distributed setting introduces significantly more noise than the centralized one, the accuracy of downstream tasks remains comparable, with errors staying within the same order of magnitude rather than degrading proportionally to the noise.

Our contributions are summarized as follows:
\begin{itemize}
\item We propose HeteroFedSyn, the first differentially private tabular data synthesis framework for federated settings with heterogeneous data.
\item Within HeteroFedSyn, we first introduce an $l_2$-based dependency metric and develop the synthesis algorithm FedPrivSyn. We then propose AdaFedPrivSyn, which incorporates an adaptive mechanism to reduce redundancy and improve efficiency in marginal selection.
\item We conduct extensive experiments across diverse downstream tasks, validating the effectiveness and practicality of our approach.
\end{itemize}


\begin{table}[t]
  \footnotesize
  \centering
  \setlength{\abovecaptionskip}{1ex} 
  \caption{Important Notations}
  \begin{tabular}{c|c}
      \toprule
      \textbf{Variable}&\textbf{Description}\\
      \hline
       $c$& Number of participants\\
       \hline
       $c_i$& The $i^{\text{th}}$ participants\\
       \hline
       $D_i$& Local dataset held by $c_i$\\
       \hline
       $\hat D$& Global synthetic dataset\\
      \hline
      $n_i$& Number of samples held by each participant\\
      \hline
      $n$& Number of total samples $n=\sum_{i=1}^c n_i$\\
      \hline
      $d$& Number of dimensions of each sample\\
      \hline
      $k$& Projection dimension of $2$-way marginal\\
      \hline
      $M_a$/$\hat M_a$& 1-way marginal/noisy marginal of the $a^{\text{th}}$ attribute\\
      \hline
      $M_{a, b}$/$\hat M_{a, b}$& \makecell{2-way marginal/noisy marginal of\\ the $a^{\text{th}}$ and $b^{\text{th}}$ attributes}\\
      \hline
      $s_a,s_b$& Length of 1-way marginal of the $a^{\text{th}}$/$b^{\text{th}}$ attribute\\
      \hline
      $\text{InDif2}_{a,b}$& \makecell{Metric indicating the correlation\\ between the $a^{\text{th}}$ and $b^{\text{th}}$ attributes}\\
      \bottomrule
  \end{tabular}
  \label{tab:notations}
  \vspace{-3ex}
\end{table}

\section{Preliminaries}
\subsection{Problem Statement}
In this paper, we study privacy-preserving data synthesis in the horizontal federated setting.
Assume there are $c$ participants, each holding a local dataset $D_i$ with $n_i$ samples.
These datasets share the same attributes but contain records from different users and exhibit distributional bias, making any single $D_i$ statistically uninformative.
Moreover, each participant $c_i$ seeks to protect the privacy of individuals in $D_i$.

An untrusted server aims to generate a synthetic dataset $\hat D$ for downstream tasks (e.g., machine learning or range queries) by leveraging the global statistical characteristics of all local datasets.
The entire data-sharing process must ensure privacy protection.
All key notations are summarized in \autoref{tab:notations}.
Except for local datasets and marginals, parameters such as the number of participants, the size of each $D_i$, and the domain size of each attribute are assumed to be commonly known to all parties and the server.

\subsection{Privacy Definitions}
We utilize Differential Privacy techniques~\cite{DBLP:conf/icalp/Dwork06} to provide privacy protection for sensitive data, which prevents any sample from being inferred from the statistical results by limiting its influence on the final statistical outcomes.

\noindent\textbf{Differential Privacy.}
First, we present the definition of standard differential privacy as follows.

\begin{definition}\label{def:differential-privacy} (\emph{($\epsilon$, $\delta$)-Differential Privacy}).
	An algorithm $\mathcal{M}$ satisfies ($\epsilon,\delta$)-DP,where $\epsilon, \delta \geq 0$, if and only if for any neighboring datasets $D$ and $D'$ that differ in one element, 
	and any possible outputs $\mathcal{R} \subseteq Range(\mathcal{M})$, we have  
	\[\Pr \left[ \mathcal{M}(D) \in \mathcal{R}\right] \leq e^{\epsilon} \Pr \left[ \mathcal{M}(D') \in \mathcal{R} \right] + \delta.\]
\end{definition}
Regarding neighboring datasets $D$ and $D'$, we consider $D$ to have one more sample than $D'$ or one less sample than $D'$ in this paper.
Here, $\epsilon$ is called the \emph{privacy budget}. 
The smaller $\epsilon$ means that the outputs of $\mathcal{M}$ on $D$ and $D'$ are more similar, and thus the provided privacy guarantee is stronger. 
$\delta$ is generally interpreted as $\mathcal{M}$ not satisfying $\epsilon$-DP with probability $\delta$.

\noindent\textbf{Zero-Concentrated DP.}
Complex algorithms usually require multiple privacy compositions.
In this paper, we achieve a more concise and tighter privacy composition by converting the privacy guarantee to Zero-Concentrated Differential Privacy (zCDP).
The definition of zCDP is as follows.

\begin{definition}\label{def:zero-differential-privacy}(Zero-Concentrated Differential Privacy (zCDP)). A randomized mechanism $\mathcal{M}$ satisfies $\rho$-zCDP with parameter $\rho > 0$ if, for all neighboring datasets $D$ and $D'$ differing in a single element, and for all $\alpha > 1$, we have:
\[
D_{\alpha}(\mathcal{M}(D) \| \mathcal{M}(D')) \leq \rho \alpha,
\]
\end{definition}
\noindent where $D_{\alpha}(\mathcal{M}(D) \| \mathcal{M}(D'))$ is the $\alpha$-Rényi divergence between the distributions of $\mathcal{M}(D)$ and $\mathcal{M}(D')$.

\noindent\textbf{Gaussian Mechanism.}
The Gaussian mechanism is commonly used to achieve $(\epsilon, \delta)$-DP.
We use it as a building block mechanism in our design.

\begin{definition}(Gaussian Mechanism).
Given a function \( f: \mathcal{D} \to \mathbb{R}^d \), the \emph{Gaussian Mechanism} adds Gaussian noise to ensure differential privacy. Specifically, the Gaussian Mechanism is defined as:
\[
\mathcal{M}(D) = f(D) + \mathcal{N}(0, \sigma^2 I),
\]
\end{definition}  
\noindent where \( \mathcal{N}(0, \sigma^2 I) \) represents a multivariate Gaussian distribution with mean zero and covariance matrix \( \sigma^2 I \). The noise scale \( \sigma \) is chosen based on the sensitivity of \( f \) and the desired privacy parameters.
The \( \ell_2 \)-sensitivity of \( f \) is given by:
\[
\Delta_f = \max_{D, D'} \| f(D) - f(D') \|_2,
\]
\noindent where \( D \) and \( D' \) are neighboring datasets differing in a single element. 
To satisfy $(\epsilon, \delta)$-DP, \( \sigma \) is typically set to:
\[
\sigma = \frac{\Delta_f\sqrt{2\ln(1.25/\delta)}}{\epsilon}.
\]
To achieve zCDP with the parameter \( \rho \)~\cite{DBLP:conf/tcc/BunS16}, \( \sigma \) is set to:
\[
\sigma = \frac{\Delta_f}{\sqrt{2\rho}}.
\]
 
\noindent\textbf{Privacy Composition.}
We use zCDP for privacy composition.
zCDP has a clean composition property, which is additive under composition.
The formal definition is as follows.

\begin{definition}(Composition of zCDP~\cite{DBLP:conf/tcc/BunS16}) 
\label{def:composition}
For a sequence of mechanisms \(\mathcal{M}_i\) satisfying \(\rho_i\)-zCDP for \(i = 1, 2, \dots, k\), their composition $\mathcal{M} = (\mathcal{M}_1, \mathcal{M}_2, \dots, \mathcal{M}_k)$ satisfies $\sum_{i=1}^{k} \rho_i\text{-zCDP}$.
\end{definition}

Under the standard \((\epsilon, \delta)\)-DP framework, we can use the following conversion to interpret zCDP guarantees.
\begin{definition}(zCDP to $(\epsilon, \delta)$-DP)
If a mechanism \(\mathcal{M}\) satisfies \(\rho\)-zCDP, then it also satisfies \((\epsilon, \delta)\)-DP for any \(\delta > 0\) and:
\[
\epsilon \leq \rho + 2\sqrt{\rho \log(1/\delta)}.
\]
\end{definition}

\noindent\textbf{Post-Processing Property.}
The noisy results processed by the DP mechanism are robust, and their privacy guarantee remains intact after further analytical exploration.
\begin{definition}
\label{def:post-process}
    If a mechanism \(\mathcal{M}: \mathcal{D} \to \mathcal{R}\) satisfies \((\epsilon, \delta)\)-DP, then for any  function \(g: \mathcal{R} \to \mathcal{R'}\), the mechanism \( g \circ \mathcal{M} \) remains \((\epsilon, \delta)\)-DP.
\end{definition}

\subsection{PrivSyn}
PrivSyn is a classical method proposed by Zhang et al.~\cite{DBLP:conf/uss/00010LH0H0021} for generating differentially private tabular datasets under the centralized setting.
It generates tabular data with similar statistical characteristics by capturing low-dimensional information, i.e., $2$-way marginals, from the table and adding noise.
The key challenge in this process is that the number of attributes $d$ in the dataset is typically large.
Releasing all $2$-way marginals, which amount to $d(d-1)/2$, would be prohibitively large, resulting in statistical characteristics being overwhelmed by noise.
Therefore, allocating a portion of the privacy budget to select the most relevant attributes' marginals is essential.
We summarize the key algorithmic process of PrivSyn into the following three steps:
\begin{itemize}[leftmargin=*]
    \item \textbf{Dependency Measurement.}
    To identify the more ``valuable'' $2$-way marginals, PrivSyn introduces a metric called InDif to evaluate the dependency between any two attributes. For any two attributes $a$ and $b$, InDif calculates the $\ell_1$ distance between the $2$-way marginal $M_{a,b}$ and $2$-way marginal generated assuming independence $M_a\times M_b$, i.e.,
    \[\text{InDif}_{a,b}=\vert M_{a,b}-M_a\times M_b\vert.\]
    According to the composition theory (\autoref{def:composition}), all $m=d(d-1)/2$ InDif scores are published with Gaussian noise, i.e., $\sigma^2=\Delta_{InDif}m/(2\rho')$, where $\Delta_{InDif}=4$.
    \item \textbf{Marginal Selection.}
    Essentially, the process of marginal selection makes trade-off between noise error and dependency error.
    If a $2$-way marginal $(a, b)$ is selected, it incurs a noise error $\psi_{a,b}$; if not selected, it introduces a dependency error $\phi_{a,b}$.
    PrivSyn utilizes $\text{InDif}_{a,b}$ as the dependency error and solves this optimization problem by proposing a greedy method:
    \begin{equation}
    \label{equ:privsyn_select}
    \min \sum_{z\in Z}[\psi_z x_z+\phi_z(1-x_z)], x_z\in\{0, 1\},
    \end{equation}
    where $Z$ represents the set of all $2$-way marginals.
    Then, the selected $2$-way marginals are released for the subsequent data synthesis after adding noise. 
    \item \textbf{Data Synthesis.}
    PrivSyn constructs a flow graph and iteratively fits the noisy marginals by duplicating and replacing values in a randomly initialized dataset.
    This method is called GUM.
    To accelerate convergence, GUM only fits noisy $2$-way marginals and generates independent attribute of samples in one step based on noisy $1$-way marginals after convergence.
    
\end{itemize}

\begin{figure*}[ht]
    \centering
    \includegraphics[width=0.8\linewidth]{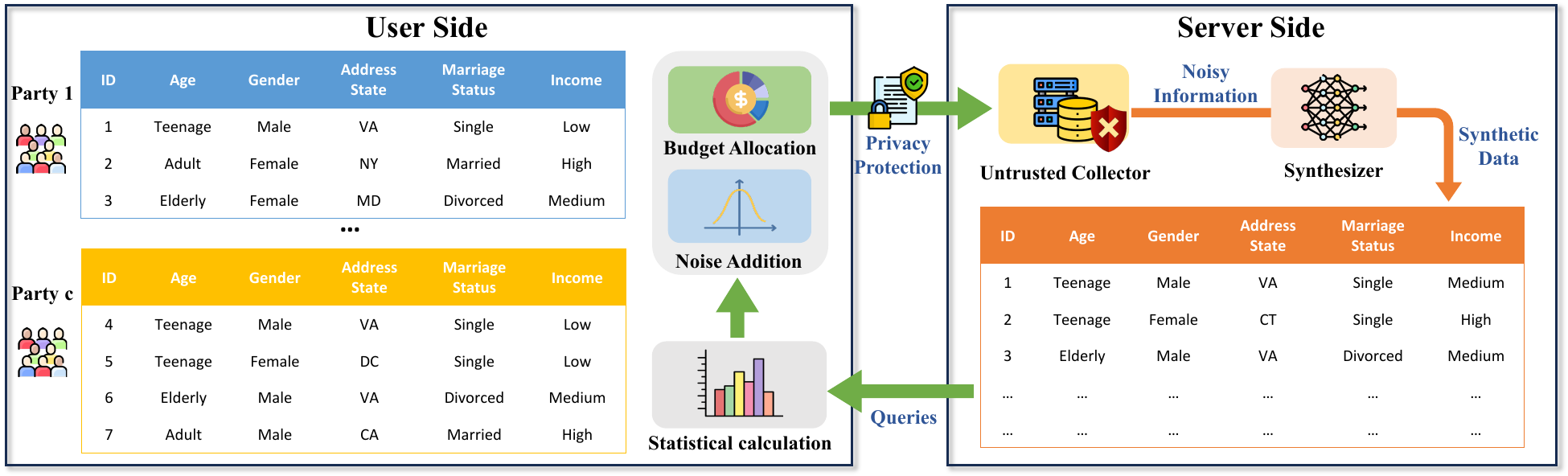}
    \caption{Overview of the horizontal federated data synthesis framework. On the user side, multiple participants hold different data sources with the same attributes. They extract information, add noise, and share it with the server, synthesizing a publishable dataset with similar statistical characteristics to the entire private data.}
    \label{fig:overview}
    \vspace{-1em}
\end{figure*}

\section{HeteroFedSyn}
In this section, we first construct a general horizontal federated tabular data synthesis framework.
HeteroFedSyn follows this framework and incorporates the algorithms for each of its components.

\subsection{Overview}
The overview of the horizontal federated data synthesis framework is shown in \autoref{fig:overview}, which contains user side and server side.
\begin{itemize}[leftmargin=*]
    \item \textbf{User Side.}
    The user side consists of $c$ clients, where private tabular data containing different samples with the same attributes are stored locally on each client (assuming the datasets have already been aligned using PSI technology~\cite{DBLP:conf/cans/CristofaroGT12, DBLP:conf/asiacrypt/KolesnikovRT019}).
    The user side never sends raw data to the server, instead, it extracts the statistical information by the \emph{Statistical Calculation} module, adds noise by the \emph{Noise Addition} module, and then transmits the noisy messages to the server.
    All privacy budget consumption is strictly controlled by the \emph{budget allocation} module, which follows the privacy composition theory (\autoref{def:composition}).
    \item \textbf{Server Side.}
    The \emph{Collector} on the server side gathers the noisy information sent by the clients and forwards it to the \emph{synthesizer} for fitting.
    This process may involve multiple rounds of interaction with the user side.
    Additionally, we assume that the collector honestly follows the algorithms but is curious and may attempt to infer user privacy.
    The synthetic dataset can be released for any downstream tasks, with the privacy guarantee ensured by the post-processing property (\autoref{def:post-process}).
\end{itemize}

\begin{figure*}[t]
    \centering
    \includegraphics[width=0.8\linewidth]{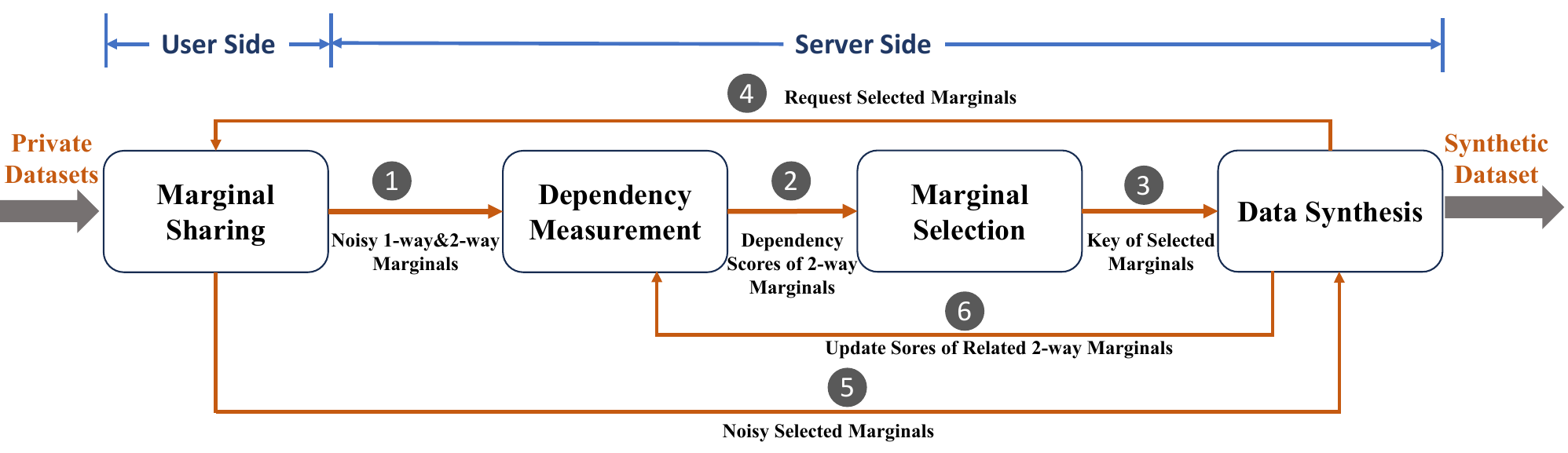}
    \caption{Workflow of HeteroFedSyn. The user side shares noisy $1$-way and $2$-way marginals with the server for dependency measurement. Based on dependency scores, the server selects $2$-way marginals for data synthesis. During the adaptive marginal selection process, HeteroFedSyn continuously updates dependency scores to guide subsequent marginal selection.}
    \label{fig:workflow}
    \vspace{-1em}
\end{figure*}

\subsection{Workflow of HeteroFedSyn}
We design HeteroFedSyn within the structure of the horizontal federated data synthesis framework.
The entire execution process consists of the following four main steps.

\begin{itemize}[leftmargin=*]
    \item \textbf{Marginal Sharing.}
    Different from the centralized setting, where all analysis and computation are performed on the raw dataset before release, the server-side analysis can only be based on the noisy information shared by the user side in the federated setting.
    At the beginning of the HeteroFedSyn, each participant on the user side needs to compute all $1$-way marginals and $2$-way marginals of their local dataset, which represent the probability distribution of individual attributes and the joint probability distribution of any two attributes. After adding noise, these marginals are sent to the server for aggregation to obtain an unbiased overall data distribution, which is crucial for the subsequent steps. Specifically, HeteroFedSyn utilizes a random projection matrix to compress all $2$-way marginals into a $k$-dimensional vector, avoiding the excessive communication overhead caused by transmitting the $2$-way marginals.
    \item \textbf{Dependency Measurement.}
    Next, HeteroFedSyn faces the challenge of measuring the dependency between attributes based on the noisy $1$-way and $2$-way marginals.
    HeteroFedSyn adopts PrivSyn's dependency evaluation metric, but modifies it from the $\ell_1$ norm to the $\ell_2$ norm.
    By leveraging the property of random projection matrix and performing debiasing operations, the server can obtain dependency scores for any pair of attributes.
    \item \textbf{Marginal Selection.}
    After obtaining the dependency scores for all $2$-way marginals, HeteroFedSyn selects the most ``valuable'' $2$-way marginals, which are then required from the user side after adding noise. This process filters out a large number of independent or weakly correlated $2$-way marginals, thereby preventing unnecessary privacy budget consumption.
    HeteroFedSyn includes two synthesis algorithms: FedPrivSyn and AdaFedPrivSyn, each employing a distinct marginal selection strategy.
    Unlike PrivSyn's static selection method, AdaFedPrivSyn adopts an adaptive approach for marginal selection.
    It continuously updates the initially computed dependency scores based on the currently selected marginals, which can further prevent wasting the privacy budget on redundant $2$-way marginals.
    \item \textbf{Data Synthesis.}
    Data synthesis is integrated throughout the entire workflow of HeteroFedSyn, rather than being performed only at the final stage after marginal selection, as in PrivSyn.
    With each $2$-way marginal selected, HeteroFedSyn performs data synthesis and adjusts the dependency scores of the relevant $2$-way marginals based on the synthesized dataset.
\end{itemize}

The workflow of HeteroFedSyn is summarized in \autoref{fig:workflow}.
\revision{Each participant first computes all $1$-way and $2$-way marginals on its local data and sends their compressed and noisy versions to the server, which uses only a small portion of the overall privacy budget (Step 1).
The server aggregates these statistics to obtain global marginals and computes dependency scores for all attribute pairs (Step 2).
Based on these scores, the server greedily selects the most informative $2$-way marginals that best capture attribute dependencies in the dataset (Step 3), and requests them from the clients (Step 4).
The clients then use the remaining majority of the privacy budget to perturb the selected $2$-way marginals and send them to the server for data synthesis (Step 5).
In the adaptive setting, the synthesized data is further used to update the dependency scores, which are fed back to guide subsequent marginal selection, enabling a dynamic selection process (Step 6).
In the following section, we provide a detailed description of the four main building blocks in HeteroFedSyn, and explain how the privacy budget is carefully managed throughout the entire process.}

\section{Building Blocks}

\subsection{Marginal Sharing}
\revision{As the first building block in the workflow of HeteroFedSyn, marginal sharing enables the server to obtain an initial understanding of global attribute distributions while preserving local data privacy.}
In HeteroFedSyn, marginal sharing on the user side occurs in two stages.
In the first stage, participants share all $1$-way and $2$-way marginals of their local datasets with the server for attribute dependency analysis.
In the second stage, after the server gains an understanding of the global dataset characteristics, it requests additional marginal information from the participants.
The first stage of marginal sharing consumes a small part of the privacy budget, reserving most of the privacy budget for the second stage to provide more accurate information for valuable marginal information.
For second-stage sharing, the user side only need to add Gaussian noise to the marginals requested by the server.
This section focuses on designing the initial stage of data sharing, which is essential but may suffer from excessive noise and communication overhead.

The marginal sharing algorithm executed by each participant is shown in \autoref{alg:marginal_share}.
All participants locally compute all $1$-way and $2$-way marginals and add Gaussian noise.
Each $2$-way marginal is multiplied by a random projection matrix ($d(d-1)/2$ in total) before adding noise, which is shared by all participants.
Next, we provide the technical details involved in the algorithm.

\setlength{\textfloatsep}{1ex}
\begin{algorithm}[h]
\small
	\caption{Marginal Sharing (Executed by Each Participant)}
	\label{alg:marginal_share}
	\begin{algorithmic}[1]
	    \REQUIRE { $\{P_{a,b}\in \mathbb{R}^{(s_a\cdot s_b)\times k}: a,b\in\{1,..., d\}, b>a\}$, where each i.i.d entry of $P_{a,b}$ drawn from $N(0, \frac{1}{k})$ and $k$ is projection dimension, local dataset $D_{i}$.}
        \STATE {Compute all 1-way marginals $\{M_a^{c_i}: a\in \{1,..., d\}\}$.}
        \STATE {Compute all 2-way marginals $\{M_{a,b}^{c_i}: a,b\in\{1,..., d\}, b>a\}$.}
        \STATE {Send $\{n_i\cdot (M_a^{c_i}+G_a): a\in \{1,..., d\}\}$} to the Server, where $G_a\in\mathbb{R}^{1\times s_a}$ is sampled from $N^{s_a}(0, \sigma_1^2)$.
        \STATE {Send $\{z_{a,b}^{c_i}=n_i\cdot (M_{a,b}^{c_i}P_{a,b}+G_{a,b}^{c_i}): a,b\in\{1,..., d\}, b>a\}\}$ to the server, where $G_{a,b}\in\mathbb{R}^{1\times k}$ is sampled from $N^k(0, \sigma_2^2)$.}
	\end{algorithmic}
\end{algorithm}
\vspace{-2ex}
\medskip
\noindent\textbf{Integrating biased local dataset.}
Since participants' data come from different sources, their distributions may vary significantly.
Conducting attribute dependency analysis on a biased local dataset would be meaningless.
As shown in \autoref{tab:gender_age_distribution}, suppose two local datasets contain the same two attributes, Gender and Age, and their one-way marginals are both [Male, 0.5, Female, 0.5], [Teenager, 0.5, Adult, 0.5].
Since Dataset \#1 primarily consists of adult males and teenage females, while Dataset \#2 mainly contains teenage males and adult females, the InDif scores calculated from these two datasets bias significantly from that of the merged global dataset.
This example further confirms the necessity of aggregating all biased local $2$-way marginals before conducting attribute dependency analysis, as aggregating only $1$-way marginals is insufficient.

Therefore, in HeteroFedSyn, all participants need to compute $1$-way and $2$-way marginals of their local datasets in the initial stage, multiply them by the size of their local dataset, and then send them to the server for aggregation.
The server sums these values and divides them by the total number of data points across all local datasets to obtain all marginals of the unbiased global dataset.

\begin{table}[t]
\footnotesize
    \centering
     \setlength{\abovecaptionskip}{1ex} 
    \caption{Examples of significant differences between local and global marginal distributions.}
    \label{tab:gender_age_distribution}
    \begin{tabular}{l |l |l}
        \toprule
        \textbf{Independent 2-way marg.} & \textbf{Actual 2-way marg.} & \textbf{InDif} \\
        \midrule
        \multicolumn{3}{c}{\textbf{Gender and Age Distribution at Dataset \#1}} \\
        \hline
        Male, Teenager\ 0.25& Male, Teenager\ 0.10 & \multirow{4}{*}{0.6}\\
        Male, Adult\ 0.25 & Male, Adult\ 0.40\\
        Female, Teenager\ 0.25 & Female, Teenager\ 0.40\\
        Female, Adult\ 0.25 & Female, Adult\ 0.10\\
        \hline
        \multicolumn{3}{c}{\textbf{Gender and Age Distribution at Dataset \#2}} \\
        \hline
        Male, Teenager 0.25 & Male, Teenager 0.40 &\multirow{4}{*}{0.6}\\
        Male, Adult 0.25 & Male, Adult 0.10\\
        Female, Teenager 0.25 & Female, Teenager 0.10\\
        Female, Adult 0.25 & Female, Adult 0.40\\
        \hline
        \multicolumn{3}{c}{\textbf{Gender and Age Global Distribution}} \\
        \hline
        Male, Teenager 0.25 & Male, Teenager 0.25& \multirow{4}{*}{0}\\
        Male, Adult 0.25 & Male, Adult 0.25\\
        Female, Teenager 0.25 & Female, Teenager 0.25\\
        Female, Adult 0.25 & Female, Adult 0.25\\
        \bottomrule
    \end{tabular}
\end{table}

\medskip
\noindent\textbf{Compressing with Random Projection.}
Sharing all marginals may result in excessive communication overhead, especially since $2$-way marginals may have very high dimensionality.
Specifically, for a dataset with $d$ attributes, the number of $2$-way marginals is $d(d-1)/2$.
For a $2$-way marginal $M_{a,b}$, where the domain size of attribute $a$ and $b$ are $s_a$ and $s_b$, the length of $M_{a,b}$ is $s_a\cdot s_b$.
The length of $M_{a,b}$ is always on the order of the squared value of $1$-way marginals.
To alleviate communication overhead, HeteroFedSyn employs random mapping techniques to reduce the dimensionality of all $2$-way marginals.
By compressing each $2$-way marginal from length $s_a \cdot s_b$ to $k$ dimensions, it not only reduces communication load, but also lowers the amount of noise added to each marginal.

According to Johnson-Lindenstrauss Lemma~\cite{lindenstrauss1984extensions}, given a data matrix $X\in \mathbb{R}^{n\times d}$, one can construct a random projection matrix $P\in \mathbb{R}^{d\times k}$ to obtain
\begin{equation}
\label{equ:projection}
    \overline X=XP\in\mathbb{R}^{n\times k},\ k\ll d.
\end{equation}
where the entries of the random matrix $P$ are i.i.d. and typically follow the Gaussian distribution or a Gaussian-like distribution, such as uniform.
It has been proven that, as long as $k$ is not too small ($k=\Omega(\log n/\lambda_{JL}^2)$, where $\lambda_{JL}$ is a small multiplicative value), two matrices multiplied by the same mapping matrix can preserve pairwise distances with high probability.
It is widely used as a dimensionality reduction tool in tasks involving distance estimation~\cite{DBLP:journals/jmlr/TomitaBSCPFPYBM20, DBLP:journals/corr/abs-2306-01751, DBLP:journals/jpc/KenthapadiKMM13}.

In the analysis of attribute dependency, $1$-way marginals are needed for further computational operations, while $2$-way marginals are only used for distance comparisons.
Therefore, HeteroFedSyn applies the random projection technique to compress all $2$-way marginals on the user side.
Since the length of each $2$-way marginal may vary, we sample a random projection matrix for each $2$-way marginal.
Specifically, for a pair of attributes ($a$, $b$), we sample a matrix $P_{a,b}\in \mathbb{R}^{(s_a\cdot s_b)\times k}$, where each entry is independently sampled from a Gaussian distribution $N(0, \frac{1}{k})$.
Each participant $c_i$ on the user side maps the $2$-way marginal $M_{a,b}^{c_i}$ using the same random projection matrix $P_{a,b}$ following \autoref{equ:projection}. 
To aggregate 2-way marginals of the same attribute pairs across participants, all parties must use the same projection matrices. 
In HeteroFedSyn, the server should generate and synchronize the $d(d-1)/2$ matrices to all participants before the data synthesis algorithm starts running.

\medskip
\noindent\textbf{Adding Gaussian Noise to Marginals.}
All $1$-way marginals and the dimension-reduced $2$-way marginals need to be perturbed with Gaussian noise to satisfy DP guarantee.
Given the allocated privacy budget $\rho_1$, for the $1$-way marginal, we have \autoref{the:1way_privacy} to ensure the privacy guarantee.


\begin{theorem}
\label{the:1way_privacy}
(Privacy Guarantee for 1-Way Marginals). Let $\Vert M_a^{c_i}-M_a^{c_i'}\Vert_2\le \Delta_1, \forall a\in \{1,..., d\}$, $i\in\{1,..., c\}$. Then for any $\rho_1>0$, $\hat M_a$ is $\rho_1$-zCDP if $\sigma_1\ge \Delta_1\sqrt{\frac{d}{2\rho_1}}$.
\end{theorem}

Next, we provide the upper bound on the sensitivity $\Delta_1$ of the $1$-way marginal in \autoref{the:delta_1}.
\begin{theorem}
\label{the:delta_1}
(Bounding $\Delta_1$). 
Let $\Delta_1$ be the $\ell_2$-sensitivity for adding noise to the $1$-way marginals. We have $\Delta_1\le 1$.
\end{theorem}
\begin{proof}
Denote the count of items on attribute $a$ as $a_1,..., a_{s_a}$.
Assume that we add one sample to the dataset of participant $c_i$, the value of attribute $a$ belongs to the $j^{\text{th}}$ item $a_j$.
Additionally, the number of sample held by each participant $n_i$ should larger than $1$.
Then we have
\begin{align*}
&\Delta_1=\max{\Vert M_a^{c_i}- M_a^{c_i'} \Vert_2}\\
&=\max{\sqrt{\sum_{i=q, q\neq j}^{s_a}(\frac{a_q}{n_i}-\frac{a_q}{n_i+1})^2+(\frac{a_j}{n_i}-\frac{a_j+1}{n_i+1})^2}}\\
&=\max{\sqrt{\sum_{i=q, q\neq j}^{s_a}\frac{a_q^2}{n_i^2(n_i+1)^2}+\frac{(a_j-n_i)^2}{n_i^2(n_i+1)^2}}}\\
&=\max{\frac{1}{n_i(n_i+1)}\sqrt{\sum_{i=q}^{s_a} a_q^2+n_i(n_i-2a_j)}}\\
&\le \frac{\sqrt{2}}{n_i}\le 1
\end{align*}
\end{proof}

Given the allocated privacy budget $\rho_2$, we present the privacy guarantee of $2$-way marginal in \autoref{the:2way_privacy}.
\begin{theorem}
\label{the:2way_privacy}
(Privacy Guarantee for 2-Way Marginals) Let $\Vert M_{a,b}^{c_i}P_{a,b}-M_{a,b}^{c_i'}P_{a,b}\Vert_2\le \Delta_2, \forall a,b\in\{1,..., d\}, b>a, i\in\{1,..., c\}$. Then for any $\rho_2>0$, $z_{a,b}$ is $\rho_2$-zCDP if $\sigma_2\ge \Delta_2\sqrt{\frac{d(d-1)}{4\rho_2}}$.
\end{theorem}

We provide an upper bound on the sensitivity $\Delta_2$ in \autoref{the:delta_2}.
\begin{theorem}
\label{the:delta_2}
(Bounding $\Delta_2$). 
Let $\Delta_2$ be the $\ell_2$ sensitivity for adding noise to the $2$-way marginals. We have 
\[\Delta_2\le \max_{1\le i\le s_as_b}\sqrt{\sum_{j=1}^k P_{a,b}[i, j]^2}\]
\end{theorem}
\begin{proof}
Denote the count of items on attribute $a$, $b$ as $a_1,..., a_{s_a}$ and $b_1,..., b_{s_b}$, respectively.
Then the domain size of $2$-way marginal on attribute $a$ and $b$ is $s_as_b$, denote the count of items as $c_1,...,c_{s_as_b}$.
Assume we add one sample to the $i^{\text{th}}$ user's datasets, the value of attribute $a$ and $b$ belong to $i_0^\text{th}$ and $j_0^\text{th}$ items, respectively. We have
\begin{align*}
\Delta_2&=\max{\Vert M_{a,b}^{c_i}P_{a,b}-M_{a,b}^{c_i'}P_{a,b}\Vert_2}\\
&\le \max{\Vert  M_{a,b}^{c_i}-M_{a,b}^{c_i'}\Vert_2}\cdot\max_{1\le i\le s_as_b}\sqrt{\sum_{j=1}^k P[i, j]^2}\quad \text{(based on \cite{DBLP:journals/jpc/KenthapadiKMM13})} \\
&=\sqrt{\sum_{(i,j)\neq(i_0,j_0)}(\frac{c_{i,j}}{n_i}-\frac{c_{i,j}}{n_i+1})^2+(\frac{c_{i_0,j_0}}{n_i}-\frac{c_{i_0,j_0}+1}{n_i+1})^2}\\
&\quad\cdot\max_{1\le i\le s_as_b}\sqrt{\sum_{j=1}^k P[i, j]^2}
\le \max_{1\le i\le s_as_b}\sqrt{\sum_{j=1}^k P[i, j]^2}.
\end{align*}

\end{proof}

\subsection{Dependency Measurement}
\revision{Following the initial marginal sharing stage, the server performs dependency measurement on the received noisy $1$-way and $2$-way marginals to identify $2$-way marginals with strong attribute correlations.
Such strongly correlated marginals play a more critical role in capturing the underlying statistical structure of the dataset than weakly correlated ones.
By allocating a small portion of the privacy budget to this dependency analysis, HeteroFedSyn can preserve the majority of the privacy budget for accurately acquiring the most informative marginals, rather than expending it on a large number of unnecessary marginals.}

HeteroFedSyn still follows InDif, the dependency measure proposed in PrivSyn~\cite{DBLP:conf/uss/00010LH0H0021}, but with modifications.
We compute the $\ell_2$ distance between the actual $2$-way marginals and the joint distribution of marginals without considering dependencies.
\begin{equation}
\label{equ:indif_2}
    \text{InDif2}_{a,b}=\Vert M_{a,b}-M_{a}\times M_{b}\Vert_2
\end{equation}

\revision{Unlike the centralized setting, where $\text{InDif2}_{a,b}$ can be computed directly from the raw data and noise added afterward, in the federated scenario, the server can only access the noisy version of the $1$-way marginals and the noisy version of the $2$-way marginals, which have been compressed through random projection.
Here, the computation of the $\ell_2$-norm distance between marginals is essentially intended to align with the property that compressed noisy marginals preserve the $\ell_2$-norm distance.
Directly applying the formula $\hat M_{a,b} - \hat M_a \times \hat M_b$ on these noisy marginals does not yield a correct $\text{InDif2}_{a,b}$, because the noise interacts in a complex way.
Therefore, the biggest challenge in HeteroFedSyn is to extract an unbiased estimate of $\text{InDif2}_{a,b}$ from these “distorted” marginals.
In particular, simply expanding $\hat M_{a,b} - \hat M_a \times \hat M_b$ would generate terms involving the original $M_a$ and $M_b$, which the server cannot access; \autoref{the:est_dif} shows how to cleverly use only the noisy marginals to cancel out these extra terms and obtain a valid unbiased estimate.}

\begin{algorithm}[t]
\small
	\caption{Dependency Measurement}
	\label{alg:dependency_calc}
	\begin{algorithmic}[1]
	    \REQUIRE {$\{P_{a,b}\in \mathbb{R}^{(s_a\cdot s_b)\times k}: a,b\in\{1,..., d\}, b>a\}$, where each i.i.d entry of $P_{a,b}$ drawn from $N(0, \frac{1}{k})$ and $k$ is projection dimension.}
        \STATE {Compute $\{\hat M_a=\frac{1}{n}\sum_{i=1}^c n_i\cdot (M_a^{c_i}+G_a): a\in \{1,..., d\}\}$.}
        \STATE {Compute $\{z_{a,b}=\frac{1}{n}\sum_{i=1}^c z_{a,b}^{c_i}: a,b\in\{1,..., d\}, b>a\}$.}
        \STATE {Compute $\{z_{a*b}=(\hat M_a\times \hat M_b)P_{a,b}: a,b\in\{1,..., d\}, b>a\}.$}
        \STATE {Compute $\{\text{InDif2}_{a,b}^2=\big\|z_{a,b}-z_{a*b}\big\|_2^2 - \Big[
k\alpha\sigma_2^2 + \alpha\sigma_1^2\big(s_b\|\hat M_a\|_2^2+s_a\|\hat M_b\|_2^2\big)
- s_a s_b \alpha^2 \sigma_1^4
\Big]: a,b\in\{1,..., d\}, b>a\}$, $\alpha=\frac{\sum_i n_i^2}{n^2}$.}
        \ENSURE {$\{\text{InDif2}_{a,b}: a,b\in\{1,..., d\}, b>a\}$}
	\end{algorithmic}
\end{algorithm}

\revision{The detailed process for dependency measurement on the server side is shown in \autoref{alg:dependency_calc}.}
The server first needs to aggregate the noisy $1$-way marginals and the compressed noisy $2$-way marginals received from the participants.
The aggregation is based on a proportional weighting according to the amount of data each participant holds.
To align the compressed noisy $2$-way marginals, the joint distribution $M_a\times M_b$ computed from noisy $1$-way marginal must also be mapped using the same projection matrix $P_{a,b}$.
After aggregation of the noisy marginal information, the server constructs an unbiased estimator of $\text{InDif2}_{a,b}$ based on it.
\revision{Next, we present how \autoref{the:est_dif} derives an unbiased estimate of $\text{InDif2}_{a,b}$ for any attribute pair $(a,b)$ using the noisy marginals $z_{a,b}$ and $z_{a*b}$}.

\begin{theorem}
\label{the:est_dif}
(Unbiased estimation of $\text{InDif2}_{a,b}$). $\big\|z_{a,b}-z_{a*b}\big\|_2^2 - \Big[
k\alpha\sigma_2^2 + \alpha\sigma_1^2\big(s_b\|\hat M_a\|_2^2+s_a\|\hat M_b\|_2^2\big)
- s_a s_b \alpha^2 \sigma_1^4
\Big]$ is an unbiased estimator of the square of $\text{InDif2}_{a,b}$, where $\alpha=\frac{\sum_i n_i^2}{n^2}$.
\end{theorem}

\begin{proof}
We first recall the setup and notations.
Each client $c_i$ has a dataset of size $n_i$ with $n=\sum_i n_i$.
Each client sends the following to the server:
1-way marginals $n_i\big(M_a^{c_i}+G_a^{c_i}\big)$ with $G_a^{c_i}\sim\mathcal N(0,\sigma_1^2 I_{s_a})$;
2-way marginals $z_{a,b}^{c_i}=n_i\big(M_{a,b}^{c_i}P_{a,b}+G_{a,b}^{c_i}\big)$ with $G_{a,b}^{c_i}\sim\mathcal N(0,\sigma_2^2 I_k)$.

Define $M_a = \frac{1}{n}\sum_i n_i M_a^{c_i}$, $G_a = \frac{1}{n}\sum_i n_i G_a^{c_i}$, $G_{a,b} = \frac{1}{n}\sum_i n_i G_{a,b}^{c_i}$, and $\alpha=\frac{\sum_i n_i^2}{n^2}$.
Then we have
$$G_a\sim\mathcal N(0,\alpha\sigma_1^2 I_{s_a}),$$
$$G_{a,b}\sim\mathcal N(0,\alpha\sigma_2^2 I_k).$$

Server aggregates 
$$\hat M_a=\frac{1}{n}\sum_i n_i(M_a^{c_i}+G_a^{c_i})=M_a+G_a,$$ 
$$z_{a,b}=\frac1n\sum_i z_{a,b}^{c_i}=M_{a,b}P_{a,b}+G_{a,b}.$$
We then analyze $E\big[\|z_{a,b}-z_{a*b} \|_2^2\big]$, where
\begin{align*}
&z_{a,b}-z_{a*b}=(M_{a,b}-M_a\times M_b)P_{a,b}\\
&+\Big(G_{a,b}-(M_a\times G_b)P_{a,b}-(G_a\times M_b)P_{a,b}-(G_a\times G_b)P_{a,b}\Big). 
\end{align*}

Since $G_a$, $G_b$, and $G_{a,b}$ are independent, mean-zero Gaussians and independent of $P_{a,b}$, we have
\begin{align*}
&\mathbb E\big[\|z_{a,b}-z_{a*b} \|_2^2\big]=E\big[\|(M_{a,b}-M_a\times M_b)P_{a,b}\|_2^2\big]\\
&+ E\big[\|G_{a,b}\|_2^2\big] + E\big[\|(M_a\times G_b)P_{a,b}\|_2^2\big] + E\big[\|(G_a\times M_b)P_{a,b}\|_2^2\big]\\
&+ E\big[\|(G_a\times G_b)P_{a,b}\|_2^2\big].
\end{align*}
We then evaluate each variance term. We use two useful properties: $\mathbb E\|xP\|_2^2=\|x\|_2^2$ (Johnson-Lindenstrauss Lemma~\cite{lindenstrauss1984extensions}), and $\|u\times v\|_2^2=\|u\|_2^2\cdot\|v\|_2^2$. Then we have
$$
\mathbb E\|(M_{a,b}-M_a\times M_b)P_{a,b}\|_2^2=\|M_{a,b}-M_a\times M_b\|_2^2,
$$
$$
\mathbb E\|G_{a,b}\|_2^2=k\alpha\sigma_2^2,
$$
$$
\mathbb E\|(M_a\times G_b)P_{a,b}\|_2^2
=\mathbb E\|M_a\times G_b\|_2^2
=\|M_a\|_2^2\cdot \mathbb E\|G_b\|_2^2
=\|M_a\|_2^2\cdot s_b\alpha\sigma_1^2.
$$
Similarly,
$$
\mathbb E\|(G_a\times M_b)P_{a,b}\|_2^2
=\|M_b\|_2^2\cdot s_a\alpha\sigma_1^2,
$$
$$
\mathbb E\|(G_a\times G_b)P_{a,b}\|_2^2
=\mathbb E\|G_a\|_2^2\,\mathbb E\|G_b\|_2^2
=s_a s_b\alpha^2\sigma_1^4.
$$
Therefore,
\begin{align*}
&E\big[\|z_{a,b}-z_{a*b} \|_2^2\big]=\|M_{a,b}-M_a\times M_b\|_2^2\\ 
&+k\alpha\,\sigma_2^2+\alpha\sigma_1^2\big(s_b\|M_a\|_2^2+s_a\|M_b\|_2^2\big)
+ s_a s_b\,\alpha^2\sigma_1^4.   
\end{align*}
Since $\hat M_a = M_a + G_a$, we have $\mathbb E\|\hat M_a\|_2^2=\|M_a\|_2^2+s_a\alpha\sigma_1^2$. Similarly, $\mathbb E\|\hat M_b\|_2^2=\|M_b\|_2^2+s_b\alpha\sigma_1^2$. Therefore, we have
\begin{align*}
&E\big[\|z_{a,b}-z_{a*b} \|_2^2\big]
=\|M_{a,b}-M_a\times M_b\|_2^2\\
&\quad +k\alpha\,\sigma_2^2+\alpha\sigma_1^2\Big(s_b \mathbb E\|\hat M_a\|_2^2+s_a\mathbb E\|\hat M_b\|_2^2\Big)
- s_a s_b\,\alpha^2\sigma_1^4.
\end{align*}
The Unbiased estimator is 
$\widehat{\text{InDif2}}^{2}_{a,b} =
\big\|z_{a,b}-z_{a*b}\big\|_2^2 - \Big[
k\alpha\sigma_2^2 + \alpha\sigma_1^2\big(s_b\|\hat M_a\|_2^2+s_a\|\hat M_b\|_2^2\big)
- s_a s_b \alpha^2 \sigma_1^4
\Big]$ where $\alpha=\frac{\sum_i n_i^2}{n^2}$, and we have $\mathbb E\big[\widehat{\text{InDif2}}^{\,2}_{a,b}\big]
=\|M_{a,b}-M_a\times M_b\|_2^2
=\text{InDif2}_{a,b}^2.$
\end{proof}

\begin{algorithm}[t]
\small
	\caption{Adaptive Marginal Selection}
	\label{alg:marginal_select}
	\begin{algorithmic}[1]
	    \REQUIRE {Number of attributes $d$, set of $2$-way marginals $Z=\{z_1,...,z_{d(d-1)/2}\}$, set of InDif2 scores $\{\phi_z^0\}_{z\in Z}$, set of noise errors $\{\psi_z\}_{z\in Z}$.}
        \STATE {Selected marginals $X\leftarrow \emptyset$, $E_0\leftarrow \sum_{z\in Z}\phi_z$, $t\leftarrow 0$}
        \WHILE {True}
        \FOR{each $2$-way marginal $z\in Z$}
        \STATE {$E_z^t=\sum_{j\in X\cup\{z\}}\psi_z+\sum_{j\in Z\setminus\{z\}}\phi_z$}
        \ENDFOR
        \STATE {$z^*\leftarrow \arg \min_{z\in Z}E_z^t$}
        \STATE {$E^t\leftarrow E_{z^*}$}
        \IF {$E^t\ge E^{t-1}$}
        \STATE {\textbf{Break}}
        \ENDIF
        \STATE {$X\leftarrow X\cup\{z^*\}$}
        \STATE {$\{\phi_z^t\}_{z\in Z}\leftarrow$ Update the set of InDif2 scores}
        \STATE {$t\leftarrow t+1$.}
        \ENDWHILE
        \ENSURE {$X$}
	\end{algorithmic}
\end{algorithm}

\begin{algorithm}[t]
\small
	\caption{Update InDif2 Scores}
	\label{alg:update_score}
	\begin{algorithmic}[1]
	    \REQUIRE {Selected marginals $X$, set of InDif2 scores $\{\phi_z\}_{z\in Z}$.}
        \STATE {Generate synthetic data $\hat D_t$ using $X$}
        \STATE {Compute 2-way marginals $\{\tilde M_{a,b}^{t}: a,b\in\{1,..., d\}, b>a\}$ from $\hat D_t$}
        \STATE {$\{\phi_z^t\}_{z\in Z}\leftarrow\{\text{InDif2}_{z}^t\leftarrow \Vert z_{a,b}-\tilde M_{a,b}^tP_{a,b} \Vert_2: a,b\in\{1,..., d\}, b>a\}\}$}
        \COMMENT{$\bigtriangledown$ \revision{Update $\text{InDif2}$ from the current synthetic data.}}
        \STATE {$Att\leftarrow \cup_{(a,b)\in X}\{a, b\}$} \COMMENT{$\bigtriangledown$ Extracts all distinct attributes contained in $X$}
        \FOR{each 2-way marginal $z\in Z$}
        \IF {$z\cap Att = \emptyset$ Or $\phi_z^t> \phi_z^{t-1}$}
        \STATE {$\phi_z^t\leftarrow \phi_z^{t-1}$} \COMMENT {$\bigtriangledown$ Only update InDif2 of $z$ related to $X$}
        \ENDIF
        \ENDFOR
        \ENSURE {$\{\phi_z^t\}_{z\in Z}$}
	\end{algorithmic}
\end{algorithm}

\subsection{Marginal Selection \& Data Synthesis}
\revision{Building upon the dependency scores $\text{InDif2}$ computed for all attribute pairs, the server proceeds to marginal selection and data synthesis.
In this stage, a subset of $2$-way marginals is selected to guide the construction of the synthetic dataset.
The selection can be directly performed using the greedy algorithm proposed in PrivSyn, as shown in \autoref{equ:privsyn_select}, which selects all marginals in a single round based on their dependency scores.
We refer to this strategy as \emph{non-adaptive} marginal selection.}

Consider the following example: In an attribute set, there exist three highly correlated attributes A, B, and C.
The $\text{InDif2}$ values for attribute pairs (A, B), (B, C), and (A, C) are higher than those of other attribute pairs, making them highly likely to be selected.
However, when (A, B) and (B, C) are already selected, the marginal information of (A, C) is implicitly contained in the combination of $M_{a,b}$ and $M_{b,c}$.
In this case, the $\text{InDif2}_{a,c}$ decreases after adjusting attributes on $M_{a,b}$ and $M_{b,c}$.
There may be other attribute pairs whose $\text{InDif2}$ is higher than updated $\text{InDif2}_{a,c}$.
Allocating the privacy budget to release other $2$-way marginals would be more beneficial for preserving the correlations in the data synthesis process.
\revision{The advantage of adaptive over static marginal selection is theoretically justified by Chen et al.~\cite{chen2025benchmarking}, which shows that $k$ already selected marginals can provide additional information for evaluating any $2$-way marginal beyond independent measurements.}
\revision{
\begin{theorem}\label{the:adapt-measure}
~\cite{chen2025benchmarking}
For any pair of attributes $(A_i, A_j)$, the KL divergence of conditional estimation is no larger than that of independent estimation:
\begin{equation}\label{eq:adapt-measure}
\mathrm{KL}\!\left(P_{ij} \,\|\, \hat P_{ij}\right)
\;\le\;
\mathrm{KL}\!\left(P_{ij} \,\|\, P_i P_j\right),
\end{equation}
where
\[
\hat P_{ij}
= \sum_{\mathbf A} P(\mathbf A)\, P(A_i\mid \mathbf A)\,P(A_j\mid \mathbf A).
\]
\end{theorem}}

\revision{Note that in HeteroFedSyn, both $\Pr[A_i \mid A_1,\ldots,A_k]$ and $\Pr[A_j \mid A_1,\ldots,A_k]$ are estimated from noisy marginals.
As a result, although \autoref{the:adapt-measure} still holds, its advantage may be attenuated by the noise.}
Based on this, we also propose an adaptive marginal selection algorithm, with the detailed process shown in \autoref{alg:marginal_select}.
Following the greedy strategy from PrivSyn, we denote the set of all $d(d-1)/2$ number of $2$ way marginals as $Z$.
For each $2$-way marginal $z$ in the set, we represent its $\text{InDif2}$ as $\phi_z$ and the noise error introduced by its release as $\psi_z$. 
In the initial step, none of the $2$-way marginals have been selected, and all errors come from $\text{InDif2}$, denoted as $\sum_{z\in Z}\phi_z$.
In each iteration, we traverse all $2$-way marginals and select the one that can most reduce the current total error $E^t$.
Intuitively, the selection process involves continuously evaluating whether the noise error introduced by releasing each $2$-way marginal is smaller than the error caused by ignoring its dependency relationship.
The key point in our adaptive marginal selection is that we update $\text{InDif2}$ of the related $2$-way marginals accordingly after adding a $2$-way marginal in each iteration.  

We show how to update $\text{InDif2}$ scores in \autoref{alg:update_score}.
In the $t$-th iteration, the server requests the newly selected $2$-way marginal from the user side (only one new $2$-way marginal is added per round), and then it executes the data synthesis algorithm based on the currently selected $2$-way marginals to obtain $\hat D_t$.
From this synthetic dataset, we can obtain all the current $2$-way marginals.
The accuracy of these $2$-way marginals can reflect the quality of the selected marginals.
\revision{We use $\text{InDif2}$ score computed from the current $2$-way marginal as the updated $\text{InDif2}$.
This step is reflected in Line 3 of the \autoref{alg:update_score}.}
If a $2$-way marginal shares an attribute with the newly added marginal, it is likely to be affected in the previous examples.
In this case, we update its $\text{InDif2}$ score.
Moreover, to avoid inaccurate marginal selection and generation from harming the accuracy of $\text{InDif2}$, we update it only when the updated $\text{InDif2}$ is smaller than the original $\text{InDif2}$.
\revision{This step is reflected in Line 5-9 of the \autoref{alg:update_score}.}

Note that, before conducting data synthesis, it is necessary to handle the attributes uncovered by the selected $2$-way marginals.
We refer to these attributes as isolated attributes.
For these attributes, we directly release their $1$-way marginals to provide information for data generation.
In this paper, we treat the data synthesis process as a black box.
We can directly invoke PrivSyn's GUM algorithm~\cite{DBLP:conf/uss/00010LH0H0021}, which iteratively replicates and replaces data points to ensure that a randomly generated dataset aligns with all the released marginals.
We can also employ other data synthesis algorithms, such as Genetic Algorithm~\cite{DBLP:journals/corr/abs-1712-06567}, Relaxed Projection~\cite{DBLP:conf/icml/AydoreBKKM0S21}, Generative Network~\cite{DBLP:conf/nips/LiuVW21}, and PGM~\cite{DBLP:conf/icml/McKennaSM19}.
The choice of synthesis algorithm does not interfere with the performance of the design in this paper.

\begin{algorithm}[t]
\small
	\caption{Privacy Budget Allocation}
	\label{alg:budget_allocation}
	\begin{algorithmic}[1]
	    \REQUIRE {Local Dataset $\{D_i\}_{i\in c}$, total privacy budget $\rho$, \revision{ratio $q$ with $q<1/3$.}}
        \ENSURE {Synthetic dataset $\hat D$}
        \STATE {Share all $1$-way marginals of $\{D_i\}_{i\in c}$ with $\rho_1=q \cdot\rho$.}
        \STATE {Share all $2$-way marginals of $\{D_i\}_{i\in c}$ with $\rho_2=q \cdot\rho$.}
        \STATE {Dependency scores measurement with no privacy loss.}
        \STATE {Adaptive Marginal Selection with no privacy loss.}
        \STATE {Update InDif2 scores with no privacy loss.}
        \STATE {Share selected $2$-way marginals of $\{D_i\}_{i\in c}$ with $\rho_3=(1-2q)\rho$.}
        \STATE {Handle isolated marginals with no privacy loss.}
        \STATE {Construct $\hat D$ using data synthesis algorithm with no privacy loss.}
	\end{algorithmic}
\end{algorithm}

\subsection{Privacy Budget Allocation}
\label{subsec: budget_alloc}
\revision{Privacy budget allocation serves as a fundamental building block that underpins the entire execution of HeteroFedSyn.
As multiple rounds of communication take place between the participants and the server, the privacy budget must be carefully allocated across different stages to ensure a strict end-to-end privacy guarantee.
We present the detailed privacy budget allocation in \autoref{alg:budget_allocation}.}

In this paper, the privacy boundary is at the participants' side.
Thus, all privacy budget consumption occurs during the data-sharing process on the user side.
All computations on the server do not introduce additional privacy loss due to the post-processing property of DP (\autoref{def:post-process}).
In the initial stage, when sharing all $1$-way and $2$-way marginals for dependency measurement, users allocate $q$ of the total privacy budget, respectively.
After selecting the important $2$-way marginals on the server side, users allocate the remaining $1-2q$ of the privacy budget to share these selected marginals.
\revision{To ensure that the majority of the privacy budget is reserved for the selected $2$-way marginals, we require $q < 1/3$, so that $1 - 2q$ dominates the total budget allocation.}
The privacy budget allocation strategy is empirical. 
Specifically, we can follow the allocation strategy in PrivSyn and set $q=10\%$.
Considering the high noise sensitivity in a distributed setting, $q$ can also be moderately increased to ensure the accuracy of the selected marginals.
We will investigate the impact of privacy budget allocation in the experiment section.

For attributes uncovered by selected $2$-way marginals, we directly use $1$-way marginals shared in the initial stage. 
Besides, each selected 2-way marginal in \autoref{alg:marginal_select} must be immediately shared with the server to update the synthetic dataset. However, since the total number of selected marginals is unknown in advance, the privacy budget cannot be pre-allocated for each sharing step.
To address this, we assume that no more than one-third of all possible 2-way marginals will be selected. This ratio serves as a hyperparameter. Based on this assumption, each time a 2-way marginal is shared, it consumes $d(d-1)q\rho/6$ of the total privacy budget.

\section{Complexity Analysis}
\revision{In this section, we analyze the computation and communication complexity of HeteroFedSyn.}

\revision{\noindent\textbf{Computation Complexity.}
Under HeteroFedSyn, the client-side computation is lightweight.
Each client computes all $1$-way and $2$-way marginals over its local dataset with complexities $O(n_i d)$ and $O(n_i d^2)$, where $n_i$ denotes the local dataset size.
Adding noise to a $1$-way marginal costs $O(s_a)$, while perturbing a $2$-way marginal $M_{a,b}$ costs $O(s_a s_b)$.
Since $2$-way marginals are perturbed twice (initial sharing and after server-side selection), the total noise cost is $O(2 s_a s_b)$.
Assuming a uniform average domain size $s$, the overall client-side complexity is
$O\!\left(d(n_i+s)+d^2(n_i+2s^2)\right)$.}

\revision{On the server side, aggregating all received marginals incurs a cost of $O(nd+nd^2)$, where $n$ is the total number of records from all clients.
Computing $\text{InDif2}$ scores for all attribute pairs requires $O(d^4)$ time.
The subsequent greedy marginal selection takes $O(d^2)$ in the worst case.
Finally, the GUM-based data synthesis algorithm has complexity $O(T d^2 n_{\text{syn}})$~\cite{chen2025benchmarking}, where $T$ is the number of iterations and $n_{\text{syn}}$ is the synthetic dataset size.
Thus, the overall server-side complexity is
$O(d^4 + d^2(n + T n_{\text{syn}}) + nd)$,
dominated by $O(d^4)$ when $d > \sqrt{n + T n_{\text{syn}}}$.
For AdaFedPrivSyn with dynamic marginal selection, the worst-case additional cost is $O(d^6 + d^4 T n_{\text{syn}})$.
In practice, by fixing the number of data synthesis reruns to a constant, the overall complexity remains dominated by $O(d^4)$.
In our experiments, we update InDif2 scores for every 10 newly selected marginals, which substantially reduces the computational cost.}

\revision{\noindent\textbf{Communication Complexity.}
Each client participates in only two communication rounds: initialization and post-selection.
Although adaptive marginal selection may involve multiple transmissions, the total communication cost remains the same as in the non-adaptive setting.
In the worst case, when all $2$-way marginals are selected, the client-to-server communication complexity is $O(2 d^2 k^2 + d k)$, where $k \ll s^2$ is the compressed marginal length.
In practice, the number of selected marginals is typically much smaller than $d^2$, depending on inter-attribute correlations.}

\revision{On the server side, it only needs to notify clients of selected $2$-way marginals by transmitting attribute index pairs.
Thus, the server-to-client communication cost is $O(d^2)$ in the worst case.}

\section{Evaluation}
\revision{In this section, we evaluate the performance of the HeteroFedSyn framework, which comprises two algorithms: FedPrivSyn and AdaFedPrivSyn. Our evaluation focuses on three aspects:
(1) the similarity between the synthetic datasets generated by HeteroFedSyn and the original datasets, as well as their utility for various downstream tasks;
(2) a comparison between HeteroFedSyn and existing approaches. Since no prior work exists under the same setting, we primarily compare against PrivSyn~\cite{DBLP:conf/uss/00010LH0H0021} from the centralized setting, along with two na\"ive distributed baselines; and
(3) the impact of varying the number of users, data distribution strategies, and key parameters on the performance of HeteroFedSyn.}

\subsection{Setup}
\noindent\textbf{Datasets.}
We evaluate our method using five real-world datasets.
All datasets are multi-dimensional, containing both numerical and categorical attributes.
We summarize all datasets in \autoref{tab:datasets}.

\begin{table}[h]
\footnotesize
\centering
\setlength{\abovecaptionskip}{1ex} 
 \setlength{\belowcaptionskip}{-1ex}
\caption{Datasets used for evaluation.}
\label{tab:datasets}
\begin{tabular}{l|cccc}
    \toprule
    \textbf{Name} & \textbf{\#Records} & \textbf{\#Attr} & \textbf{\#Num} & \textbf{\#Cat} \\ 
    \midrule
    Adult~\cite{adult} & 32,562  & 15 & 6 & 9 \\ 
    Abalone~\cite{abalone} & 4,177 & 9 & 8 & 1 \\
    Obesity~\cite{obesity} & 2,112& 17& 8& 9 \\
    Insurance~\cite{insurance} & 1,339 & 7 & 3 & 4 \\
    Shoppers~\cite{shoppers} & 12,331 & 18 & 10 & 8 \\
    \bottomrule
\end{tabular}
\end{table}

Prior to the experiments, all datasets are preprocessed to enable efficient computation of marginals for each attribute.
Specifically, we directly encode all categorical attributes without altering their domain sizes.
For numerical attributes, we divide their ranges into $100$ bins and mapped all data accordingly.
\begin{figure*}[t]
    \centering
    \includegraphics[width=0.7\linewidth]{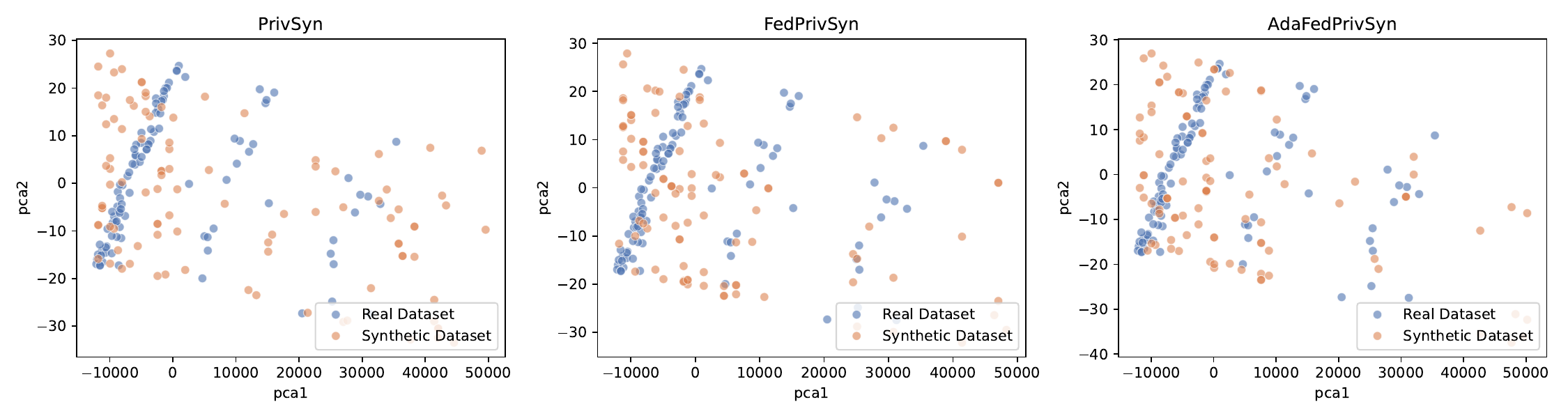}
    \vspace{-1ex}
    \caption{\revision{Comparison of the original dataset and synthetic datasets generated by PrivSyn, FedPrivSyn, and AdaFedPrivSyn in a two-dimensional PCA space. The results are evaluated on Insurance dataset, where $\epsilon=5$, number of participants is $c=2$, and $40\%$ privacy budget is allocated for releasing selected $2$-way marginals.}}
    \label{fig:pca} 
    \vspace{-1em}
\end{figure*}

\begin{figure*}[t]
    \centering
    \includegraphics[width=0.7\linewidth]{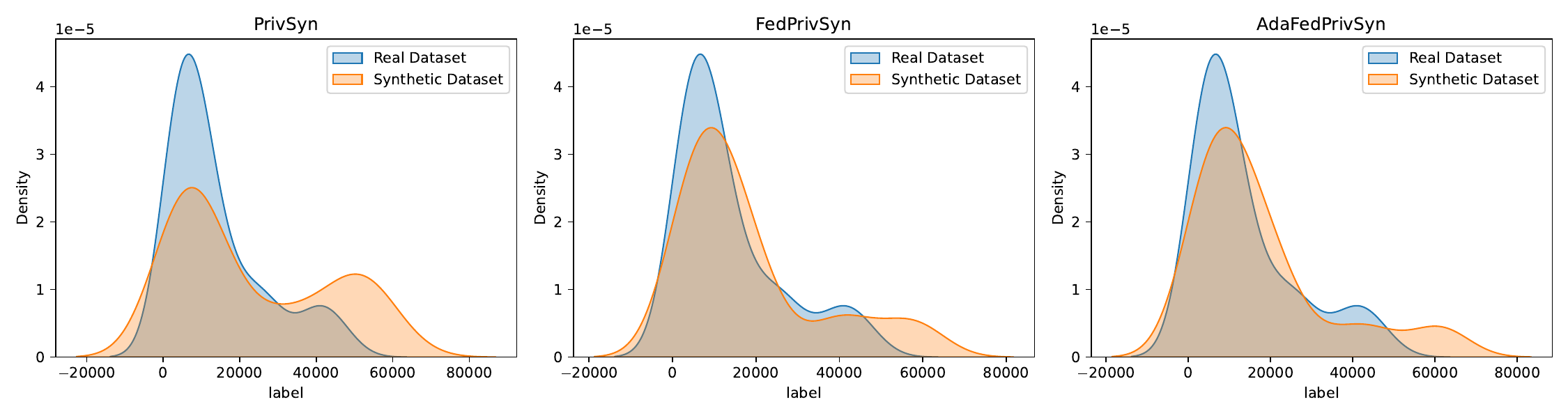}
    \vspace{-1ex}
    \caption{\revision{Comparison of the single-attribute distribution between the synthetic datasets generated by PrivSyn, FedPrivSyn, and AdaFedPrivSyn and the original dataset. The results are evaluated on the \emph{Label} attribute of the Insurance Dataset, where $\epsilon=5$, number of participants is $c=2$, and $40\%$ privacy budget is allocated for releasing selected $2$-way marginals.}}
    \label{fig:attribute_dist}
    \vspace{-1em}
\end{figure*}
\noindent\textbf{Downstream Tasks and Metrics.}
Du et al.~\cite{du2024systematic} propose a comprehensive and systematic evaluation framework for data synthesis algorithms.
Our evaluation follows their framework and specifically includes the following aspects.
\begin{itemize}[leftmargin=*]
    \item \emph{Utility in Range Query Tasks.} For each evaluation, we generate $1,000$ random queries.
    For discrete attributes, query key sets are sampled at random, while for continuous attributes, query ranges are defined by randomly generated lower and upper bounds.
    Each query returns the proportion of records within the specified range relative to the entire dataset.
    The query error is measured as the average difference between the results obtained on the synthetic dataset and those on the original dataset for all queries.
    Formally, let $\mathcal{A}$ denote the synthesis algorithm being evaluated, $D_{syn}$ the synthetic dataset generated by $\mathcal{A}$, $D_{org}$ the original dataset, and $R$ the query set.
    The query error is then indicated as 
    \begin{equation}
        Err_{rq}=\frac{1}{R}\sum_{r \in R} \left| q_r(D_{syn}) - q_r(D_{org}) \right|.
    \end{equation}
    \item \emph{Wasserstein-based Fidelity.}
    The Wasserstein distance quantifies the minimum cost required to transform one distribution into another.
    In our evaluation, we compute the Wasserstein distance between the $2$-way marginals of the synthetic dataset and those of the original dataset. Formally, let $P$ denote the $2$-way marginal distribution derived from the synthetic dataset $D_{syn}$, and $Q$ the corresponding distribution from the original dataset $D_{org}$.
    The fidelity of a synthesis algorithm $A$ is then defined as follows:
    \begin{equation}
        Fidelity(A):=\mathbb{E}_{P\sim D_{syn}, Q\sim D_{org}}[W(P, Q)].
    \end{equation}
    For the details of Wasserstein distance $W(P, Q)$, we refer the reader to~\cite{du2024systematic}.
    \item \emph{Accuracy in Machine Learning Tasks.}
    We apply the synthetic datasets to train three machine learning models: Random Forest, Multilayer Perceptron (MLP), and XGBoost~\cite{DBLP:conf/kdd/ChenG16}.
    For dataset with categorical labels, we train classification models; for datasets with numerical labels, we train regression models.
    The entire synthetic dataset is used for model training, while $20\%$ of original dataset is reserved for validating model accuracy.
    Finally, we report the average performance of the synthetic dataset across three models, using F1 score as the evaluation metric for classification tasks and Root Mean Square Error (RMSE) for regression tasks.
\end{itemize}

All synthetic datasets are generated with the same size as the corresponding original datasets.
Given the high computational cost of the entire process, we repeat the data synthesis and evaluation process five times to reduce randomness.

\noindent\textbf{Experimental Setting.}
The development and evaluation of our experiments are conducted on top of an open-source framework~\cite{synmeter}.
All the algorithms are implemented in Python 3.10 and all the experiments are conducted on a server running Ubuntu 22.04.1, equipped with E5-2620 v4 2.10GHz processors and 128GB of memory. 

\begin{table*}[!h]
\footnotesize
    \centering
    \setlength{\tabcolsep}{4.5pt}
    \setlength{\abovecaptionskip}{1ex} 
    \setlength{\belowcaptionskip}{0ex}
    \caption{\revision{Overall performance of the synthetic datasets generated by different methods across three different downstream tasks on multiple datasets. We evaluate under three levels of privacy protection strength. The arrows next to $\epsilon$ indicators denote whether a higher or lower result is better. For the distributed methods, the number of participants is $c=5$, the projected dimension $k=10$, and $80\%$ privacy budget is allocated for releasing selected $2$-way marginals.}}
    \resizebox{0.99\textwidth}{!}{
    \begin{tabular}{l|ccc|ccc|ccc|ccc|ccc}
    \toprule 
        \textbf{Dataset} & \multicolumn{3}{c|}{\textbf{Adult}} & \multicolumn{3}{c|}{\textbf{Abalone}} & \multicolumn{3}{c|}{\textbf{Obesity}} & \multicolumn{3}{c|}{\textbf{Insurance}} & \multicolumn{3}{c}{\textbf{Shoppers}}\\ 
        \cmidrule{1-16}
        \textbf{Query Error} & $\varepsilon = 0.2\downarrow$ & $\varepsilon = 1\downarrow$ & $\varepsilon = 5\downarrow$ & $\varepsilon = 0.2\downarrow$ & $\varepsilon = 1\downarrow$ & $\varepsilon = 5\downarrow$ & $\varepsilon = 0.2\downarrow$ & $\varepsilon = 1\downarrow$ & $\varepsilon = 5\downarrow$ & $\varepsilon = 0.2\downarrow$ & $\varepsilon = 1\downarrow$ & $\varepsilon = 5\downarrow$ & $\varepsilon = 0.2\downarrow$ & $\varepsilon = 1\downarrow$ & $\varepsilon = 5\downarrow$ \\ 
        \midrule
        PrivSyn & $0.011$ & $0.004$ & $0.003$ & $0.058$ & $0.057$ & $0.046$ & $0.085$ & $0.052$ & $0.027$ & $0.059$ & $0.044$ & $0.026$ & $0.039$ & $0.019$ & $0.007$ \\
        FedPrivSyn-allMarg & $0.035$ & $0.031$ &$0.016$ & $0.061$ & $0.058$ & $0.056$ & $0.099$ & $0.093$ & $0.076$ & $0.057$ & $0.053$ & $0.046$ & $0.054$ & $0.051$ & $0.037$ \\
        FedPrivSyn-RandMarg & $0.043$ & $0.035$ & $0.022$ & $0.057$ & $0.060$ & $0.059$ & $0.095$ & $0.096$ & $0.073$ & $0.061$ & $0.060$ & $0.048$ & $0.054$ & $0.050$ & $0.037$\\
        \hline
        FedPrivSyn-w/o RP & $0.035$ & $0.018$ & $0.007$ & $0.060$ & $0.059$ & $0.056$ & $0.105$ & $0.077$ & $0.051$ & $\textbf{0.056}$ & $\textbf{0.053}$ & $0.031$ & $0.052$ & $0.041$ & $0.018$\\
        FedPrivSyn & $0.018$ & $0.018$ & $\textbf{0.005}$ & $0.058$ & $\textbf{0.056}$ & $0.054$ & $0.093$ & $0.083$ & $0.049$ & $0.083$ & $0.055$ & $0.038$ & $0.051$ & $0.039$ & $0.016$\\
        AdaFedPrivSyn & $\textbf{0.017}$ & $\textbf{0.009}$ & $0.006$ & $\textbf{0.057}$ & $0.058$ & $\textbf{0.044}$ & $\textbf{0.091}$ & $\textbf{0.076}$ & $\textbf{0.030}$ & $0.094$ & $0.062$ & $\textbf{0.028}$ & $\textbf{0.048}$ & $\textbf{0.020}$ & $\textbf{0.007}$\\
        \cmidrule{1-16}
        \textbf{Fidelity Error} & $\varepsilon = 0.2\downarrow$ & $\varepsilon = 1\downarrow$ & $\varepsilon = 5\downarrow$ & $\varepsilon = 0.2\downarrow$ & $\varepsilon = 1\downarrow$ & $\varepsilon = 5\downarrow$ & $\varepsilon = 0.2\downarrow$ & $\varepsilon = 1\downarrow$ & $\varepsilon = 5\downarrow$ & $\varepsilon = 0.2\downarrow$ & $\varepsilon = 1\downarrow$ & $\varepsilon = 5\downarrow$ & $\varepsilon = 0.2\downarrow$ & $\varepsilon = 1\downarrow$ & $\varepsilon = 5\downarrow$ \\ 
        \midrule
        PrivSyn & $0.293$ & $0.086$ & $0.036$ & $0.451$ & $0.430$ & $0.360$ & $0.337$ & $0.263$ & $0.148$ & $0.475$ & $0.361$ & $0.194$ & $0.622$ & $0.317$ & $0.116$ \\
        FedPrivSyn-allMarg & $0.589$ & $0.476$ & $0.230$ & $0.462$ & $0.462$ & $0.449$ & $0.352$ & $0.352$ & $0.330$ & $0.359$ & $0.327$ & $0.275$ & $0.802$ & $0.757$ & $0.575$\\
        FedPrivSyn-RandMarg & $0.587$ & $0.463$ & $0.326$ & $0.461$ & $0.460$ & $0.457$ & $0.340$ & $0.342$ & $0.305$ & $0.345$ & $0.352$ & $0.316$ & $0.795$ & $0.758$ & $0.584$\\
        \hline
        FedPrivSyn-w/o RP & $0.542$ & $0.280$ & $0.112$ & $0.453$ & $\textbf{0.451}$ & $0.431$ & $0.523$ & $\textbf{0.290}$ & $0.222$ & $\textbf{0.279}$ & $\textbf{0.312}$ & $0.240$ & $0.775$ & $0.635$ & $0.318$\\
        FedPrivSyn & $0.230$ & $0.281$ & $0.061$ & $0.454$ & $0.454$ & $0.425$ & $0.345$ & $0.299$ & $0.220$ & $0.667$ & $0.385$ & $0.267$ & $0.757$ & $0.602$ & $0.300$ \\
        AdaFedPrivSyn & $\textbf{0.158}$ & $\textbf{0.097}$ & $\textbf{0.034}$ & $\textbf{0.411}$ & $0.458$ & $\textbf{0.342}$ & $0.542$ & $0.391$ & $\textbf{0.137}$ & $0.587$ & $0.525$ & $\textbf{0.158}$ & $\textbf{0.39}4$ & $\textbf{0.295}$ &  $\textbf{0.104}$\\
        \cmidrule{1-16}
        \textbf{ML Efficiency} & $\varepsilon = 0.2\uparrow$ & $\varepsilon = 1\uparrow$ & $\varepsilon = 5\uparrow$ & $\varepsilon = 0.2\downarrow$ & $\varepsilon = 1\downarrow$ & $\varepsilon = 5\downarrow$ & $\varepsilon = 0.2\uparrow$ & $\varepsilon = 1\uparrow$ & $\varepsilon = 5\uparrow$ & $\varepsilon = 0.2\downarrow$ & $\varepsilon = 1\downarrow$ & $\varepsilon = 5\downarrow$ & $\varepsilon = 0.2\uparrow$ & $\varepsilon = 1\uparrow$ & $\varepsilon = 5\uparrow$ \\ 
        \midrule
        PrivSyn & $0.652$ & $0.743$ & $0.752$ & $3.157$ & $3.065$ & $3.313$ & $0.103$ & $0.135$ & $0.188$ & $149.746$ & $143.610$ & $131.762$ & $0.775$ & $0.726$ & $0.724$ \\
        FedPrivSyn-allMarg & $0.687$ & $0.687$ & $0.763$ & $3.326$ & $3.326$ & $3.355$ & $0.107$ & $0.110$ & $0.131$ & $132.409$ & $130.797$ & $130.222$ & $0.524$ & $0.773$ & $0.769$\\
        FedPrivSyn-RandMarg & $0.620$ & $0.651$ & $0.685$ & $3.225$ & $2.667$ & $3.375$ & $0.108$ & $0.152$ & $0.136$ & $144.673$ & $144.531$ & $142.742$ & $0.655$ & $0.778$ & $0.766$\\
        \hline
        FedPrivSyn-w/o RP & $0.666$ & $0.668$ & $0.681$ & $3.287$ & $3.097$ & $3.418$ & $0.029$ & $0.117$ & $\textbf{0.145}$ & $142.630$ & $137.846$ & $131.009$ & $0.743$ & $0.772$ & $0.741$\\
        FedPrivSyn & $0.648$ & $\textbf{0.736}$ & $0.752$ & $\textbf{3.052}$ & $2.714$ & $3.192$ & $\textbf{0.118}$ & $0.131$ & $0.134$ & $150.798$ & $136.469$ & $130.937$ & $\textbf{0.749}$ & $0.775$ & $0.712$\\
        AdaFedPrivSyn & $0.609$ & $0.718$ & $0.728$ & $3.247$ & $3.211$ & $\textbf{3.006}$ & $0.029$ & $0.113$ & $0.128$ & $151.314$ & $144.364$ & $137.539$ & $0.706$ & $0.707$ & $\textbf{0.782}$\\
    \bottomrule
    \end{tabular}
    }
    \label{tab:entire_result}
    \vspace{-3ex}
\end{table*}

\subsection{Analysis of Experimental Results}
\noindent\textbf{Performance of Synthesis Algorithms.}
In this part, we aim to demonstrate the utility of the synthetic datasets generated by the proposed synthesis algorithms from different perspectives, across various tasks and datasets.

\emph{(1) Distributions comparison of the synthetic and original datasets.}
\autoref{fig:pca} illustrates the comparison of distributional similarity between the original sensitive dataset and the datasets generated by the centralized PrivSyn, as well as our proposed methods, FedPrivSyn and AdaFedPrivSyn.
For ease of visualization, we use the relatively small Insurance dataset.
All attributes are projected into a two-dimensional PCA space, with blue points representing the original dataset and orange points representing the synthetic datasets.
We set the overall privacy budget to $\epsilon=5$.
For FedPrivSyn and AdaFedPrivSyn, $60\%$ of the total privacy budget is allocated to share the $1$-way and $2$-way marginals, while the remaining $40\%$ of the budget is used to release the selected $2$-way marginals.

From a visual perspective, we can observe that, similar to PrivSyn, although the synthetic data points do not exactly overlap with the original points, their overall distribution is roughly similar to that of the original dataset.
The DP noise inevitably increases the variance of the synthetic datasets, and the distributed noise addition further introduces noise that is participant number times larger than in the centralized setting.
Nevertheless, \autoref{fig:pca} shows that the datasets generated by FedPrivSyn and AdaFedPrivSyn are overall comparable to those produced by PrivSyn.
The \emph{Label} attribute in the Insurance dataset is numerical, and we are interested in whether the synthetic datasets preserve the distributed characteristics of this attribute.
This is important for downstream tasks such as training regression models.
\autoref{fig:attribute_dist} compares the \emph{Label} distributions of datasets generated by PrivSyn, FedPrivSyn, and AdaFedPrivSyn with the original dataset.
Although the variance of the distribution increases, the synthetic datasets still remain the main characteristics of the original distribution.
Notably, AdaFedPrivSyn preserves the peak information of the distribution more prominently.

\revision{\emph{(2) Performance comparison of HeteroFedSyn with baseline methods.}
\autoref{tab:entire_result} compares synthesis algorithms across downstream tasks, datasets, and privacy levels ($\epsilon=0.2$, $1$, and $5$). 
We compare HeteroFedSyn with three baselines. PrivSyn~\cite{DBLP:conf/uss/00010LH0H0021} represents a classic method under the centralized setting. FedPrivSyn-allMarg adds noise to all 2-way marginals without marginal selection and directly uses them for data synthesis. FedPrivSyn-RandMarg randomly selects a random number of 2-way marginals, perturbs them, and uses the resulting marginals for synthesis. 
We also evaluate FedPrivSyn-w/o RP, which removes random projection for marginal compression. We do not treat it as a baseline, since random projection is mainly introduced to reduce communication overhead rather than to improve synthesis accuracy, and is therefore considered part of our method.
For fair comparison, FedPrivSyn and AdaFedPrivSyn follow the same privacy budget allocation as PrivSyn, with $20\%$ for marginal selection and $80\%$ for perturbing 2-way marginals. We set the number of participants to $c=5$ and the random projection dimension to $k=10$.}

\revision{For range queries, we report the error of 2-way marginals, where lower values indicate better performance. Fidelity is measured using the Wasserstein distance between the real and synthetic datasets. We further evaluate utility on three downstream machine learning tasks: Random Forest, MLP, and XGBoost, by training models on synthetic data and validating them on a held-out subset of the original data. For the Abalone and Insurance datasets, we use regression models and report prediction error (lower is better); for the remaining datasets, we use classification models evaluated by F1 score (higher is better).
Overall, despite injecting substantially more noise due to distributed privacy constraints, the distributed methods achieve errors within the same order of magnitude as PrivSyn across all datasets and metrics. Among all methods, AdaFedPrivSyn consistently achieves the best performance on both range query and fidelity on Adult and Shoppers containing a large number of attributes.
In contrast, on datasets with fewer attributes and smaller attribute domains, such as Insurance, our methods do not exhibit a clear advantage. This is expected, since when the number of 2-way marginals is limited and their domain sizes are small, redundancy among marginals is low. In such cases, allocating part of the privacy budget to marginal selection and dimensionality reduction provides limited benefit.
We also observe that all methods exhibit similar trends on range query and fidelity metrics, while the advantages of distributed methods on downstream ML tasks are less consistent. This is likely because ML performance depends not only on marginal accuracy but also on higher-order correlations and task-specific inductive biases, which may not be fully captured by improved 2-way marginals alone. As a result, gains in query accuracy and fidelity do not always translate directly into uniform improvements in downstream learning performance.}

\noindent\textbf{Impact of Key Parameters on the Synthesis Algorithms.}
When applying algorithms in real-world scenarios, it is inevitable to configure some key parameters.
In this part, we investigate how these parameters affect the performance of the proposed algorithms.

\begin{figure}[t]
    \centering
    \includegraphics[width=0.85\linewidth]{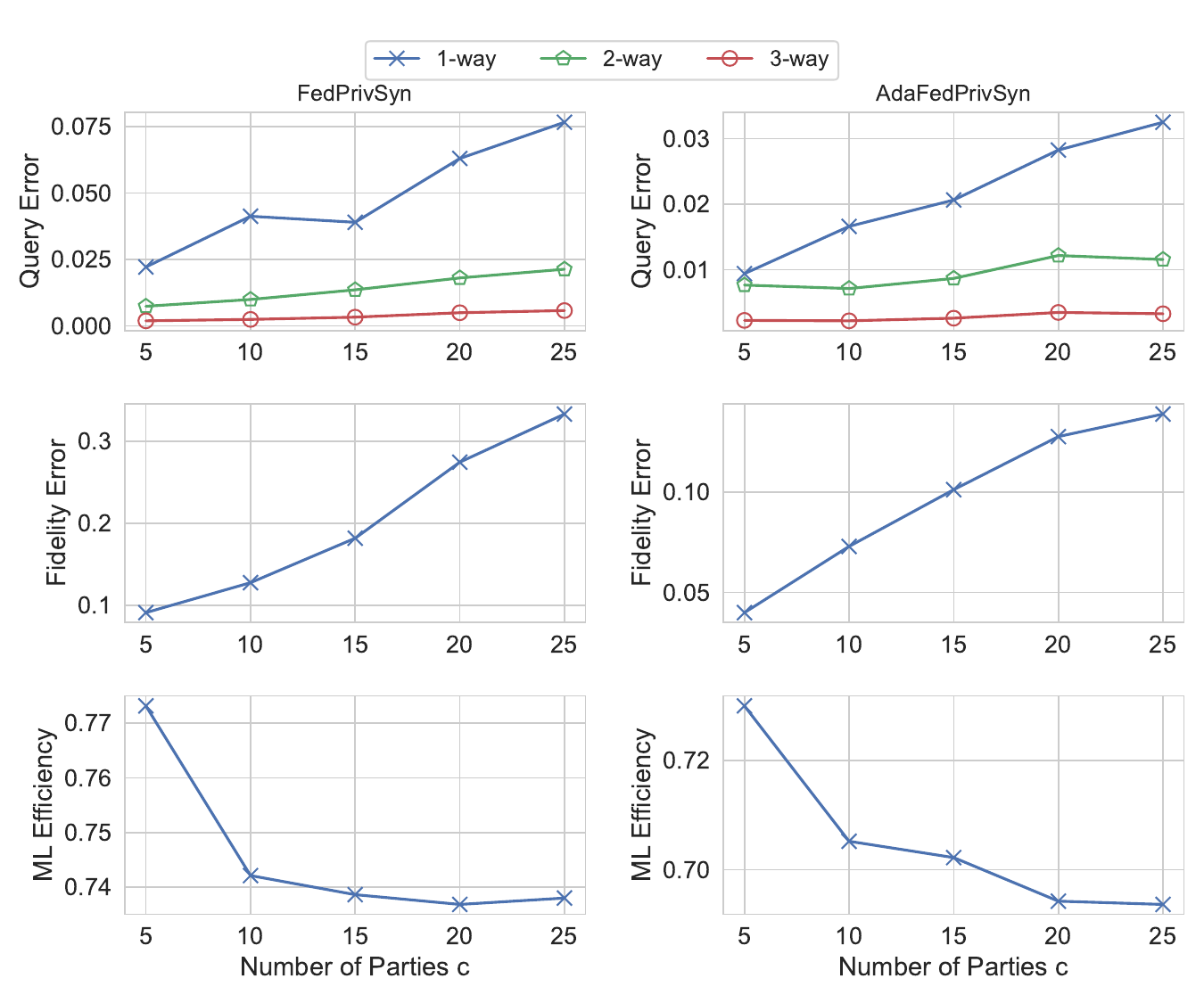}
    \vspace{-1ex}
    \caption{\revision{Impact of number of participants $c$ on FedPrivSyn and AdaFedPrivSyn across three downstream tasks on Adult dataset, where $\epsilon=5$, and $40\%$ privacy budget is allocated for releasing selected $2$-way marginals.}}
    \label{fig:vary_c_utility}
    \vspace{-1em}
\end{figure}

\emph{(1) Impact of the number of participants $c$.}
On the Adult dataset, we vary the number of participants $c$ from $5$ to $25$, assuming that data samples are uniformly distributed across all participants. 
By evaluating the range query error, fidelity error, and ML efficiency of the generated datasets, we demonstrate the impact of increasing the number of distributed participants on FedPrivSyn and AdaFedPrivSyn.
The results are shown in \autoref{fig:vary_c_utility}.
As the number of distributed participants increases, the dataset is spread across more client-sides, which means that the aggregated dataset inevitably contains more noise. 
Consequently, the performance of the synthetic datasets on the three tasks decreases as $c$ increases.
However, we observe that for range query error on 2-way and 3-way marginals, as well as for the fidelity error, the increase in errors of the synthetic datasets slows down as $c$ grows. The error growth rates of AdaFedPrivSyn consistently lower than those of FedPrivSyn.

\begin{figure}[t]
    \centering
    \includegraphics[width=0.9\linewidth]{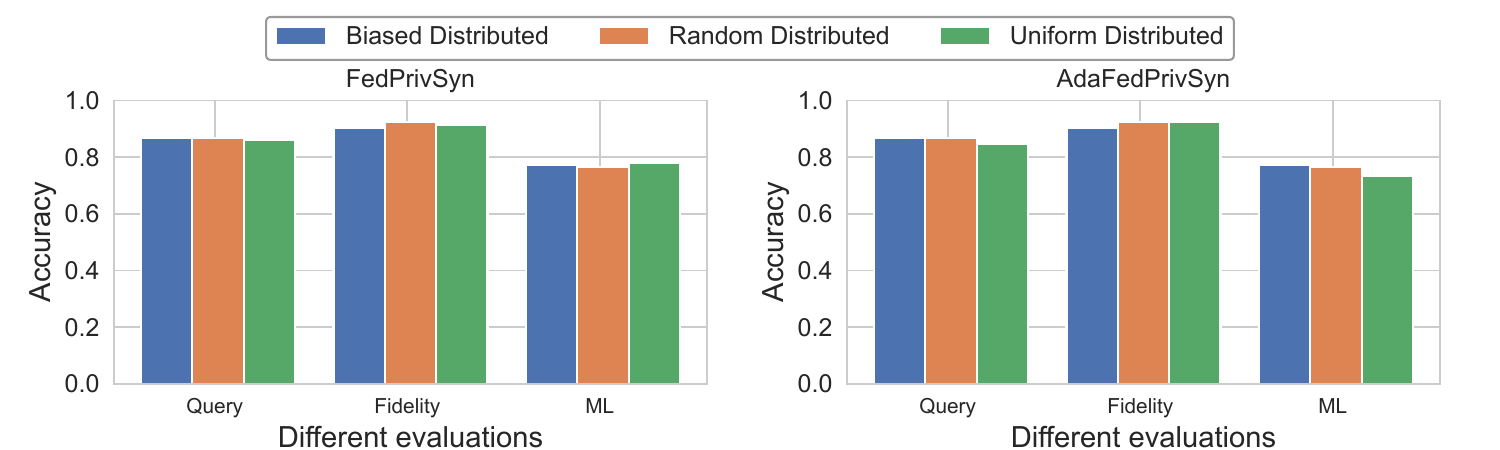}
    \vspace{-1ex}
    \caption{\revision{Impact of data distribution on FedPrivSyn and AdaFedPrivSyn across three downstream tasks on Adult dataset, where $\epsilon=5$, $c=5$, and $40\%$ privacy budget is allocated for releasing selected $2$-way marginals.}}
    \label{fig:distribution_imapct}
\end{figure}

\begin{figure*}[t]
    \centering
    \includegraphics[width=0.7\linewidth]{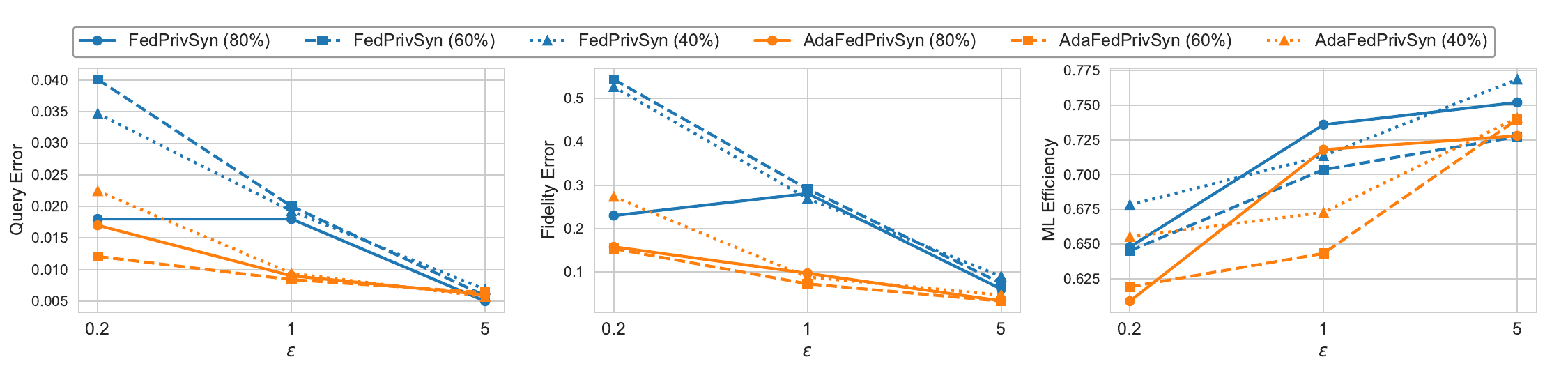}
    \vspace{-1ex}
    \caption{\revision{Impact of privacy budget allocation on FedPrivSyn and AdaFedPrivSyn across three downstream tasks on Adult dataset, with $40\%$, $60\%$, or $80\%$ of the privacy budget assigned to selected 2-way marginals and $c=5$.}}
    \label{fig:epsilon_dist}
    \vspace{-1em}
\end{figure*}

\begin{figure*}[t]
    \centering
    \includegraphics[width=0.7\linewidth]{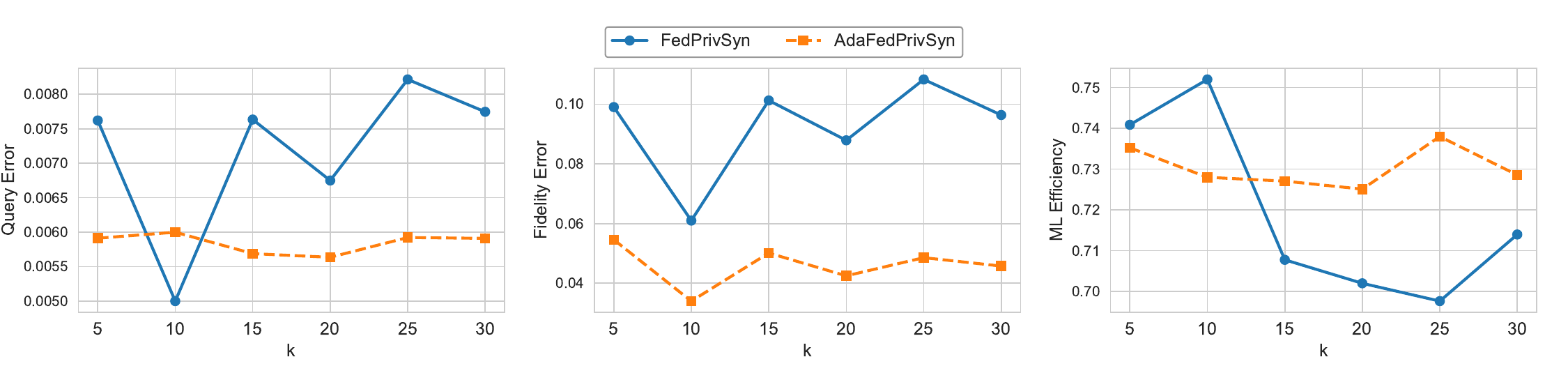}
    \vspace{-1ex}
    \caption{\revision{Impact of random projection dimension $k$ on FedPrivSyn and AdaFedPrivSyn across three downstream tasks on Adult dataset, where $\epsilon=5$, $c=5$, and $40\%$ privacy budget is allocated for releasing selected $2$-way marginals.}}
    \label{fig:k_vary}
    \vspace{-1em}
\end{figure*}

\revision{
\emph{(2) Impact of data distribution among participants.}
In our evaluation, we assume that the dataset is evenly partitioned among participants. For example, a dataset of size 1,000 distributed across five users assigns 200 samples to each. In practice, data may be unevenly distributed, with some participants holding substantially more samples or exhibiting significant distributional bias.
\autoref{fig:distribution_imapct} compares FedPrivSyn and AdaFedPrivSyn under three settings: (i) uniform data distribution across participants, (ii) random data allocation with unequal sample sizes, and (iii) distributions with pronounced bias, where participants hold data from different subpopulations (one participant may contain only high-income/low-income individuals). Since the evaluation metrics operate on different scales, we normalize them for joint visualization. The F1 score (ML Efficiency) is left unchanged, while query error and fidelity error are transformed using $e^{-V/\alpha}$, with $\alpha=0.05$ for range queries and $\alpha=1$ for fidelity. After normalization, larger values consistently indicate better performance.
As shown in \autoref{fig:distribution_imapct}, AdaFedPrivSyn and FedPrivSyn exhibit comparable performance across all three settings. This robustness arises because the noisy statistics are shared and proportionally aggregated, which effectively mitigates the impact of distributional bias. These results confirm the effectiveness of HeteroFedSyn in handling heterogeneous data distributions.}
\begin{table}[t]
\footnotesize
\centering
 \setlength{\tabcolsep}{4pt}
 \setlength{\abovecaptionskip}{1ex} 
 \setlength{\belowcaptionskip}{0ex}
\caption{\revision{Impact of the assumption on the number of selected marginals (1/4, 1/3, and 1/2 of all 2-way marginals) on AdaFedPrivSyn, where $c=5$ and $40\%$ privacy budget is allocated for releasing selected $2$-way marginals.}}
\label{tab:query_error_marginals}
\begin{tabular}{c c c c c c}
\toprule
Dataset & $\epsilon$ & Assumption & Query Error & Fidelity Error & \# Selected Marginals \\
\midrule
\multirow{4}{*}{Adult} 
& \multirow{2}{*}{5} 
  & 1/4 & 0.006 & 0.044 &  11/26\\
& & 1/3 & 0.006 & 0.034 &  35/35\\
& & 1/2 & 0.006 & 0.059 & 52/52\\
\cmidrule(lr){2-6}
& \multirow{2}{*}{1} 
  & 1/4 & 0.011 & 0.090 &  11/26\\
& & 1/3 & 0.017 & 0.158 & 35/35\\
& & 1/2 & 0.015 & 0.204 &  51/52\\
\midrule
\multirow{4}{*}{Obesity} 
& \multirow{2}{*}{5} 
  & 1/4 & 0.023 & 0.084 & 10/34 \\
& & 1/3 & 0.030 & 0.084 & 11/45 \\
& & 1/2 & 0.028 & 0.132 & 20/68\\
\cmidrule(lr){2-6}
& \multirow{2}{*}{1} 
  & 1/4 & 0.085 & 0.330 & 10/34 \\
& & 1/3 & 0.067 & 0.333 & 10/45\\
& & 1/2 & 0.093 & 0.410 & 10/68 \\
\bottomrule
\end{tabular}
\vspace{-1ex}
\end{table}

\revision{\emph{(3) Impact of privacy budget allocation.}
In PrivSyn, the privacy budget allocation is fixed: $20\%$ of the budget is used to select informative 2-way marginals, and the remaining $80\%$ is allocated to releasing the selected marginals. For the comparison reported in \autoref{tab:entire_result}, we configure FedPrivSyn and AdaFedPrivSyn using the same budget allocation to ensure fairness.
However, this allocation is not necessarily optimal. 
In \autoref{fig:epsilon_dist}, we investigate whether different privacy budget allocations have a significant impact on the quality of the synthesized data. We observe that when the privacy budget is relatively large (e.g., $\epsilon\ge 1$), different allocation strategies yield comparable performance. In contrast, under a tighter privacy budget (e.g., $\epsilon=0.2$), allocating a larger portion of the budget (e.g., $80\%$) to releasing 2-way marginals leads to noticeably better performance for both methods. When the privacy budget is sufficient, allocating more budget to marginal selection can help identify more informative 2-way marginals and further improve synthesis quality.}

\revision{\emph{(4) Impact of random projection dimension.}
In FedPrivSyn and AdaFedPrivSyn, we map marginals to a lower-dimensional space to reduce communication overhead. While adding noise in the reduced space decreases the amount of injected noise, dimensionality reduction also introduces compression error. As shown in \autoref{fig:k_vary}, varying the projection dimension $k$ from 5 to 30 leads to non-monotonic performance changes across all three metrics for both algorithms. Although an overall trend can be observed, the results exhibit noticeable fluctuations.
Notably, $k=10$ achieves consistently better performance across all metrics, suggesting a favorable balance between compression error and noise error. In practice, the optimal choice of $k$ strongly depends on the underlying data distribution. For some datasets, such as Insurance, compression can noticeably degrade the performance of HeteroFedSyn (\autoref{tab:entire_result}); however, the reduction in communication overhead is substantial.}

\revision{\emph{(5) Impact of the assumption in AdaFedPrivSyn.}
As described in \autoref{subsec: budget_alloc}, AdaFedPrivSyn repeatedly performs marginal selection and updates InDif2, which requires participants to share selected 2-way marginals in advance. Since the total number of selected marginals is unknown, we assume that at most a fixed fraction of all 2-way marginals will be chosen and allocate the privacy budget accordingly. \autoref{tab:query_error_marginals} evaluates this assumption using three bounds ($1/4$, $1/3$, and $1/2$) on two datasets under two privacy budgets. We observe that a larger bound generally leads to more selected marginals, as the reduced per-marginal budget degrades InDif2 accuracy, though this does not necessarily harm overall utility because the optimal number depends on how many informative marginals are truly needed. Finding this optimum is difficult in the federated setting: static FedPrivSyn ignores marginal dependencies, while AdaFedPrivSyn’s dynamic selection is sensitive to the assumed bound. In practice, we set the bound to $1/3$, which achieves consistently good performance (\autoref{tab:entire_result}).}

\section{Related Work}
Tabular data synthesis has a long research trajectory, primarily focused on the centralized setting.
Existing privacy-preserving works primarily follow two technical approaches.
 
A research line uses probabilistic graphical models (PGMs), which represent correlations among random variables through graph structures.
Some approaches are based on Bayesian networks, such as PrivBayes~\cite{DBLP:conf/sigmod/ZhangCPSX14} and BSG~\cite{DBLP:journals/pvldb/BindschaedlerSG17}, which approximate the joint distribution by privately selecting conditional distributions from attribute–parent pairs. However, repeated use of the exponential mechanism in this selection incurs significant computational overhead and accuracy loss under a limited privacy budget.
Other methods, including PGM~\cite{DBLP:conf/icml/McKennaSM19}, PrivMRF~\cite{cai2021data}, and JTree~\cite{DBLP:conf/kdd/ChenXZX15}, rely on Junction Tree–based Markov networks. PrivMRF and JTree construct an undirected graph by adding attribute pairs with noisy mutual information above a threshold, while PGM formulates an optimization problem to infer a distribution matching observed marginals.
Another direction selects low-dimensional marginals with strong correlations using noise-addition mechanisms, as in MST~\cite{DBLP:journals/jpc/McKennaMS21} and PrivSyn~\cite{DBLP:conf/uss/00010LH0H0021}. Their synthesis strategies differ: MST uses PGM to generate data, whereas PrivSyn employs GUM to iteratively modify a randomly initialized dataset via duplication and replacement.
Finally, AIM~\cite{mckenna2022aim} adopts a workload-aware approach, iteratively refining marginals by selecting and updating poorly answered queries using the exponential mechanism.

Another line of work leverages deep generative models. Early studies focused on Generative Adversarial Networks (GANs)~\cite{xie2018differentially, torkzadehmahani2019dp, jordon2018pate}, which generate data under differential privacy using DP-SGD~\cite{abadi2016deep} by injecting noise into gradients during training. However, GANs were originally designed for continuous data, such as images, and struggle to handle datasets containing mixed numerical, categorical, and ordinal variables. They also have difficulty capturing the full distribution, as the generator produces limited variation.
More recent work explores diffusion models~\cite{truda2023generating, sattarov2024differentially} and Large Language Models (LLMs)~\cite{tran2024differentially} under differential privacy. These approaches, however, typically require large training datasets.

Few studies address data synthesis in federated settings. Most existing methods rely on GANs~\cite{zhao2021fed, fang2022dp, duan2022ht}. The only statistical approach applies PrivBayes to discrete data~\cite{su2016differentially}. However, constructing the Bayesian network requires numerous communication rounds between the server and participants, and its performance is sensitive to the network initialization. As shown in~\cite{DBLP:conf/uss/00010LH0H0021}, PrivSyn outperforms PrivBayes. Moreover, since the implementation of~\cite{su2016differentially} is unavailable, we do not include it in our comparison.

\section{Conclusion}
In this work, we proposed HeteroFedSyn, a differentially private data synthesis framework for the federated setting. Within this framework, we tackled the key challenge of marginal selection without direct access to distributed data, and further introduced an adaptive optimization strategy to enhance selection efficiency.
While this fills an important gap in DP synthesis, it also highlights that handling high-dimensional statistics in a distributed manner causes noise to accumulate across every step under a limited privacy budget.
As a result, even well-designed optimizations may be easily overshadowed.
Looking ahead, we believe further improvements can be achieved not only through algorithmic refinements but also by leveraging public knowledge to reduce privacy cost. We hope HeteroFedSyn serves as a foundation for more practical and scalable privacy-preserving data sharing in federated scenarios.

\section{Acknowledgments}
We thank the anonymous reviewers for their valuable comments and suggestions. Cong Shen is supported by NSF grants AST-2132700, ECCS-2332060, CNS-2313110, and ECCS-2143559. Jing Yang and Fengyu Gao are supported by NSF grants ECCS-2531023 and CNS-2531789.
\bibliographystyle{ACM-Reference-Format}
\bibliography{sample-base}

\appendix

\end{document}